\def\UseBibLatex{1}
\ifx\GenSICOMPVer\undefined%
\newcommand{\SICOMP}[1]{}%
\newcommand{\NotSICOMP}[1]{#1}%
\else
\newcommand{\SICOMP}[1]{#1}%
\newcommand{\NotSICOMP}[1]{}%
\fi

\makeatletter
\def\input@path{{sicomp/}{../sicomp/}} 
\makeatother

\SICOMP{%

   \documentclass[review,onefignum,onetabnum]{siamart190516}
   \let\UseBibLatex\undefined
}%
\NotSICOMP{%
   \documentclass[12pt]{article}%
}

\SICOMP{%
\usepackage{lipsum}
\usepackage{amsfonts}
\usepackage{graphicx}
\usepackage{epstopdf}
\usepackage{algorithmic}
\ifpdf
  \DeclareGraphicsExtensions{.eps,.pdf,.png,.jpg}
\else
  \DeclareGraphicsExtensions{.eps}
\fi

\newsiamremark{remark}{Remark}
\newsiamremark{hypothesis}{Hypothesis}
\crefname{hypothesis}{Hypothesis}{Hypotheses}
\newsiamthm{claim}{Claim}

\usepackage{amsopn}%

}

\IfFileExists{sariel_computer.sty}{\def\sarielComp{1}}{}
\ifx\sarielComp\undefined%
\newcommand{\SarielComp}[1]{}
\newcommand{\NotSarielComp}[1]{#1}%
\else
\newcommand{\SarielComp}[1]{#1}%
\newcommand{\NotSarielComp}[1]{}%
\fi
\newcommand{\IfPrinterVer}[2]{#2}%

\newcommand{\UsePackage}[1]{%
  \IfFileExists{../styles/#1.sty}{%
      \usepackage{../styles/#1}%
   }{%
      \IfFileExists{./styles/#1.sty}{%
         \usepackage{styles/#1}%
      }{%
         \usepackage{#1}%
      }%
   }%
}

\NotSICOMP{%
   \usepackage[cm]{fullpage}%
}%
\usepackage{array}%

\usepackage{amsmath}%
\usepackage{amstext}%
\usepackage{amssymb}%
\usepackage{bm}%
\usepackage{xcolor}%
\usepackage{upgreek}%
\usepackage{euscript}%
\NotSICOMP{%
   \usepackage{caption}%
}%
\usepackage[normalem]{ulem}%
\usepackage{stmaryrd}%
\usepackage{multirow}%

\SarielComp{\usepackage{sariel_colors}}%

\NotSICOMP{%
   \usepackage[amsmath,thmmarks]{ntheorem}%
   \theoremseparator{.}%
   
   \usepackage{titlesec}%
   \titlelabel{\thetitle. }%
}

\usepackage{xcolor}%
\usepackage{mleftright}%
\usepackage{xspace}%
\usepackage{scalerel}%

\DeclareFontFamily{U}{BOONDOX-calo}{\skewchar\font=45 }
\DeclareFontShape{U}{BOONDOX-calo}{m}{n}{
  <-> s*[1.05] BOONDOX-r-calo}{}
\DeclareFontShape{U}{BOONDOX-calo}{b}{n}{
  <-> s*[1.05] BOONDOX-b-calo}{}
\DeclareMathAlphabet{\mathcalb}{U}{BOONDOX-calo}{m}{n}
\SetMathAlphabet{\mathcalb}{bold}{U}{BOONDOX-calo}{b}{n}
\DeclareMathAlphabet{\mathbcalb}{U}{BOONDOX-calo}{b}{n}

\usepackage{color}

\newcommand{\Do}{{\small\bf do}\ }

\newcommand{\Return}{{\small\bf return\ }}

\newcommand{\If}{{\small\bf if}\ }
\newcommand{\Then}{{\small\bf then}\ }
\newcommand{\Else}{{\small\bf else}\ }

\newcommand{\For}{{\small\bf for}\ }

\newcommand{\xbeginlgox}{\begin{minipage}{1in}\begin{tabbing}
           \quad\=\qquad\=\qquad\=\qquad\=\qquad\=\qquad\=\qquad\=\kill}
        \newcommand{\xendlgox}{\end{tabbing}\end{minipage}}

\newenvironment{algorithmEnv}{
   \begin{tabular}{|l|}\hline\xbeginlgox\bigskip}
    {\xendlgox\\\hline\end{tabular}}

\newenvironment{program}{
   \begin{minipage}{4.0in}
   \begin{tabbing}
       \ \ \ \ \= \ \ \ \= \ \ \ \ \= \ \ \ \ \= \ \ \ \ \=
      \ \ \ \ \= \ \ \ \ \= \ \ \ \ \= \ \ \ \ \=
      \ \ \ \ \= \ \ \ \ \= \ \ \ \ \= \ \ \ \ \= \kill
}{
   \end{tabbing}
   \end{minipage}
}

\DefineNamedColor{named}{AlgorithmColor}{cmyk}{0.07,0.90,0,0.34}

\newcommand{\AlgorithmI}[1]{{%
      \textcolor[named]{AlgorithmColor}{\texttt{\bf{#1}}}%
   }}
\newcommand{\Algorithm}[1]{{%
      \AlgorithmI{#1}%
      \index{algorithm!#1@{\AlgorithmI{#1}}}%
   }}

\usepackage[inline]{enumitem}

\newlist{compactenumA}{enumerate}{5}%
\setlist[compactenumA]{topsep=0pt,itemsep=-1ex,partopsep=1ex,parsep=1ex,%
   label=(\Alph*)}%

\newlist{compactenumi}{enumerate}{5}%
\setlist[compactenumi]{topsep=0pt,itemsep=-1ex,partopsep=1ex,parsep=1ex,%
   label=(\roman*)}%

\newlist{compactenumI}{enumerate}{5}%
\setlist[compactenumI]{topsep=0pt,itemsep=-1ex,partopsep=1ex,parsep=1ex,%
   label=(\Roman*)}%

\usepackage{ifluatex}
\usepackage{ifxetex}

\ifluatex %
      \usepackage{fontspec}
      \usepackage[utf8]{luainputenc}
\else
       \ifxetex %
          \usepackage{fontspec}
       \else
          \usepackage[T1]{fontenc}
          \usepackage[utf8]{inputenc}
       \fi
\fi
\providecommand{\BibLatexMode}[1]{}
\providecommand{\BibTexMode}[1]{#1}

\ifx\UseBibLatex\undefined%
  \renewcommand{\BibLatexMode}[1]{}
  \renewcommand{\BibTexMode}[1]{#1}
\else
  \renewcommand{\BibLatexMode}[1]{#1}
  \renewcommand{\BibTexMode}[1]{}
\fi

\BibLatexMode{%
   \usepackage[bibencoding=ascii,style=alphabetic,backend=biber]{biblatex}%
   \UsePackage{sariel_biblatex}%
}

\usepackage{hyperref}%

\IfPrinterVer{%
}{%
   \hypersetup{%
      breaklinks,%
      colorlinks=true,%
      urlcolor=[rgb]{0.25,0.0,0.0},%
      linkcolor=[rgb]{0.5,0.0,0.0},%
      citecolor=[rgb]{0,0.2,0.445},%
      filecolor=[rgb]{0,0,0.4},
      anchorcolor=[rgb]={0.0,0.1,0.2}%
   }
}
\NotSICOMP{%
   \hypersetup{ocgcolorlinks}%
}

\definecolor{darkgreen}{rgb}{0,0.3,0}

\definecolor{blue25}{rgb}{0,0,0.7}
\providecommand{\emphic}[2]{%
   \textcolor{blue25}{%
      \textbf{\emph{#1}}}%
   \index{#2}}
\providecommand{\emphi}[1]{\emphic{#1}{#1}}

\NotSICOMP{%
\theoremseparator{.}%

\theoremstyle{plain}%
\newtheorem{theorem}{Theorem}[section]

\newtheorem{lemma}[theorem]{Lemma}

\newtheorem{corollary}[theorem]{Corollary}
\newtheorem{fact}[theorem]{Fact}

\theoremstyle{plain}%
\theoremheaderfont{\sf} \theorembodyfont{\upshape}%
\newtheorem*{remark:unnumbered}[theorem]{Remark}%
\newtheorem{remark}[theorem]{Remark}%
\newtheorem{definition}[theorem]{Definition}
\newtheorem{defn}[theorem]{Definition}
\newtheorem*{defn:unnumbered}[FakeCounter]{Definition}

\newtheorem{problem}[theorem]{Problem}

\newcommand{\myqedsymbol}{\rule{2mm}{2mm}}

\theoremheaderfont{\em}%
\theorembodyfont{\upshape}%
\theoremstyle{nonumberplain}%
\theoremseparator{}%
\theoremsymbol{\myqedsymbol}%
\newtheorem{proof}{Proof:}%
}

\SICOMP{%
\newsiamremark{defn}{Definition}%
\newsiamremark{fact}{Fact}%
\newsiamremark{problem}{Problem}%
}

\newcommand{\atgen}{\symbol{'100}}
\newcommand{\SarielThanks}[1]{\thanks{Department of Computer Science;
      University of Illinois; 201 N. Goodwin Avenue; Urbana, IL,
      61801, USA; {\tt sariel\atgen{}illinois.edu}; {\tt
         \url{http://sarielhp.org/}.} #1}}

\newcommand{\MitchellThanks}[1]{%
   \thanks{%
      Department of Computer Science;
      University of Illinois; 201 N. Goodwin Avenue; Urbana, IL,
      61801, USA; {\tt mfjones2\atgen{}illinois.edu}; {\tt
         \url{http://mfjones2.web.engr.illinois.edu/}.} #1}}
         
\newcommand{\TimothyThanks}[1]{\thanks{Department of Computer Science;
      University of Illinois; 201 N. Goodwin Avenue; Urbana, IL,
      61801, USA; {\tt tmc\atgen{}illinois.edu}; {\tt
         \url{http://tmc.web.engr.illinois.edu/}.} #1}}

\numberwithin{figure}{section}%
\numberwithin{table}{section}%
\numberwithin{equation}{section}%

\newcommand{\HLink}[2]{\hyperref[#2]{#1~\ref*{#2}}}
\newcommand{\HLinkSuffix}[3]{\hyperref[#2]{#1\ref*{#2}{#3}}}

\newcommand{\figlab}[1]{\label{fig:#1}}
\newcommand{\figref}[1]{\HLink{Figure}{fig:#1}}

\newcommand{\thmlab}[1]{{\label{theo:#1}}}
\newcommand{\thmref}[1]{\HLink{Theorem}{theo:#1}}

\newcommand{\corlab}[1]{\label{cor:#1}}
\newcommand{\corref}[1]{\HLink{Corollary}{cor:#1}}%

\newcommand{\seclab}[1]{\label{sec:#1}}
\newcommand{\secref}[1]{\HLink{Section}{sec:#1}}

\newcommand{\remlab}[1]{\label{rem:#1}}
\newcommand{\remref}[1]{\HLink{Remark}{rem:#1}}%

\newcommand{\problab}[1]{\label{prob:#1}}
\newcommand{\probref}[1]{\HLink{Problem}{prob:#1}}

\providecommand{\deflab}[1]{\label{def:#1}}
\newcommand{\defref}[1]{\HLink{Definition}{def:#1}}

\newcommand{\apndlab}[1]{\label{apnd:#1}}
\newcommand{\apndref}[1]{\HLink{Appendix}{apnd:#1}}

\newcommand{\itemlab}[1]{\label{item:#1}}
\newcommand{\itemref}[1]{\HLinkSuffix{}{item:#1}{}}

\newcommand{\lemlab}[1]{\label{lemma:#1}}
\newcommand{\lemref}[1]{\HLink{Lemma}{lemma:#1}}%

\newcommand{\factlab}[1]{\label{fact:#1}}
\newcommand{\factref}[1]{\HLink{Fact}{fact:#1}}

\newcommand{\tbllab}[1]{\label{table:#1}}
\newcommand{\tblref}[1]{\HLink{Table}{table:#1}}

\providecommand{\eqlab}[1]{}%
\renewcommand{\eqlab}[1]{\label{equation:#1}}
\newcommand{\Eqref}[1]{\HLinkSuffix{Eq.~(}{equation:#1}{)}}

\providecommand{\remove}[1]{}%

\newcommand{\Set}[2]{\left\{ #1 \;\middle\vert\; #2 \right\}}
\newcommand{\norm}[1]{\left\| {#1} \right\|}
\newcommand{\pth}[2][\!]{\mleft({#2}\mright)}%

\newcommand{\pbrcx}[1]{\left[ {#1} \right]}%

\newcommand{\Ex}[2][\!]{\mathop{\mathbf{E}}#1\pbrcx{#2}}

\newcommand{\ceil}[1]{\left\lceil {#1} \right\rceil}
\newcommand{\floor}[1]{\left\lfloor {#1} \right\rfloor}

\newcommand{\brc}[1]{\left\{ {#1} \right\}}
\newcommand{\cardin}[1]{\left| {#1} \right|}%
\newcommand{\abs}[1]{\left| {#1} \right|}

\renewcommand{\th}{th\xspace}

\newcommand{\tldO}{\scalerel*{\widetilde{O}}{j^2}}%
\newcommand{\polylog}{\text{polylog}}

\renewcommand{\Re}{\mathbb{R}}%

\newcommand{\etal}{et~al.\xspace~}
\providecommand{\Mh}[1]{#1}

\newcommand{\eps}{\varepsilon}
\newcommand{\R}{\mathbb{R}}
\newcommand{\N}{\mathbb{N}}

\newcommand{\PS}{\Mh{P}}%
\newcommand{\PSA}{\Mh{Q}}%
\newcommand{\QS}{\Mh{Q}}%

\newcommand{\pa}{\Mh{p}}%
\newcommand{\pq}{\Mh{q}}%
\newcommand{\pb}{\Mh{s}}%
\newcommand{\pc}{\Mh{c}}%

\newcommand{\query}{\Mh{{z}}}%

\newcommand{\uDepth}{\Mh{\mathsf{u}}}%

\newcommand{\Line}{\Mh{\ell}}%
\newcommand{\hplane}{\Mh{h}}%

\newcommand{\hp}{\Mh{h}}%
\newcommand{\hs}{\Mh{h}}%
\newcommand{\hsa}{\Mh{f}}%

\newcommand{\HS}{\Mh{H}}

\newcommand{\EPlanesX}[1]{\Mh{E}_{#1}}
\newcommand{\PlanesX}[1]{\Mh{H}_{#1}}

\newcommand{\Simp}{\Delta}

\newcommand{\BX}[1]{\partial #1}

\newcommand{\Disk}{\Mh{\mathsf{o}}}%

\newcommand{\Ball}{\Mh{\mathsf{b}}}%

\newcommand{\rad}{\Mh{\nu}}%

\newcommand{\ray}{\rho}

\newcommand{\Dual}[1]{#1^\star}

\newcommand{\Cell}{\Mh{\Delta}}%
\newcommand{\Cells}{\Xi}
\newcommand{\ZoneY}[2]{\mathcal{Z}\pth{#1, #2}}
\newcommand{\lvY}[2]{L_{#1}(#2)}
\newcommand{\TY}[2]{T_{#1}(#2)}
\newcommand{\BY}[2]{B_{#1}(#2)}

\newcommand{\VertsX}[1]{\mathsf{V}\pth{#1}}
\newcommand{\Arr}{\Mh{\mathcal{A}}}%
\newcommand{\ArrX}[1]{\Arr\pth{#1}}

\newcommand{\LP}{\Term{LP}\xspace}%
\newcommand{\Objv}{\Mh{f}}%
\newcommand{\ObjvB}{\Mh{g}}%
\newcommand{\ObjvExt}{\Mh{\mathsf{F}}}%
\newcommand{\ObjvExtX}[1]{\ObjvExt\pth{#1}}%
\newcommand{\ObjvX}[1]{\Objv\pth{#1}}
\newcommand{\Impl}{\Mh{g}}%
\newcommand{\CN}{\Mh{{H}}}%
\newcommand{\CNa}{\Mh{{B}}}%
\newcommand{\CNc}{\Mh{{C}}}%

\newcommand{\CNX}{\Mh{{X}}}%
\newcommand{\cnx}{\Mh{{x}}}%

\newcommand{\CB}{\Mh{\mathcal{Z}}}%

\newcommand{\CNA}{\Mh{\mathcal{B}}}%
\newcommand{\CNb}{\Mh{{C}}}%
\newcommand{\CNs}{\Mh{{S}}}%
\newcommand{\CNP}{\Mh{\mathcal{H}}}%
\newcommand{\CNB}{\Mh{\mathcal{C}}}%
\newcommand{\CNR}{\Mh{R}}%
\newcommand{\cons}{\Mh{h}}

\newcommand{\map}{\kappa}

\newcommand{\Basis}{\Mh{\mathcalb{b}}}%
\newcommand{\BasisB}{\Mh{\mathcalb{t}}}%
\newcommand{\CD}{\Mh{\delta}}%

\newcommand{\lpdim}{\Mh{\delta}}%
\newcommand{\sza}{\Mh{\psi}}%
\newcommand{\szb}{\Mh{\eta}}%
\newcommand{\DX}[1]{\Mh{\mathcal{D}}\pth{#1}}%

\newcommand{\DT}{\Mh{\mathcal{D}}}%

\newcommand{\IRX}[1]{\left\llbracket #1 \right\rrbracket}
\newcommand{\IntRange}[1]{\left\llbracket #1 \right\rrbracket}
\newcommand{\permut}[1]{\left\langle {#1} \right\rangle}

\DefineNamedColor{named}{AlgorithmColor}{cmyk}{0.07,0.90,0,0.34}

\providecommand{\AlgorithmI}[1]{{%
      \textcolor[named]{AlgorithmColor}{\texttt{\bf{#1}}}%
   }}
\providecommand{\Algorithm}[1]{{%
      \AlgorithmI{#1}%
      \index{algorithm!#1@{\AlgorithmI{#1}}}%
   }}

\newcommand{\violate}{\AlgorithmI{violate}\xspace}%
\newcommand{\compBasis}{\AlgorithmI{compBasis}\xspace}%
\newcommand{\solveLPT}{\AlgorithmI{solveBatchLPT}\xspace}%

\definecolor{OliveGreen}{cmyk}{0.64, 0, 0.95, 0.40 }

\newlength{\savedparindent}

\newcommand{\SaveContent}[2]{%
   \expandafter\newcommand{#1}{#2}%
}

\newcommand{\Matousek}{Matou{\v{s}}ek\xspace}%
\newcommand{\Term}[1]{\textsf{#1}}%

\newcommand{\SO}{S_\textup{odd}}
\newcommand{\SE}{S_\textup{even}}

\newcommand{\TP}{\Mh{\mathcal{T}}}%

\newcommand{\ColZ}[3]{%
   \pagestyle{empty}%
   \begin{minipage}{#3}
       \pagestyle{empty}%
       \smallskip%
       #1\\
       #2 \smallskip%
   \end{minipage}%
}%
\newcommand{\ColY}[2]{%
   \ColZ{#1}{#2}{2.6cm}%
}
\providecommand{\TPDF}[2]{\texorpdfstring{#1}{#2}}

\newcommand{\cross}{\mathsf{d}}
\newcommand{\TCBX}[1]{\Mh{T_{\mathrm{c{}b}}}\pth{#1}}%
\newcommand{\TLPX}[1]{\Mh{T_{\mathrm{l{}p}}}\pth{#1}}%
\newcommand{\DotProdY}[2]{\permut{#1, #2}}

\newcommand{\HSup}{\Mh{\HS_\uparrow}}%
\newcommand{\DHSup}{\Mh{\HS^\star_\uparrow}}%
\newcommand{\HSdown}{\Mh{\HS_\downarrow}}%
\newcommand{\DHSdown}{\Mh{\HS^\star_\downarrow}}%

\newcommand{\Simplex}{\Mh{\nabla}}%

\newcommand{\Instance}{\Mh{\mathcal{I}}}%

\newcommand{\rankX}[1]{\Mh{\mathrm{r}}\pth{#1}}%
\newcommand{\rankMinX}[1]{\Mh{\zeta}\pth{#1}}

\newcommand{\levelX}[1]{\Mh{\mathrm{level}}\pth{#1}}%
\newcommand{\Bundle}{\Mh{\Upsilon}}%

\newcommand{\solveLPType}       {\Algorithm{solveLPType}\xspace}

\newcommand{\MST}{\textsf{MST}}

\BibLatexMode{%
   \bibliography{yolk}
}

\title{%
   Optimal Algorithms for Geometric Centers and Depth%
   \thanks{%
      This paper is a merge of two conference papers that were
      published sixteen years apart. The first paper \cite{c-oramt-04}
      appeared in SODA 2004, and the second paper \cite{hj-fagc-20}
      (which can be viewed as an applications paper of the first
      paper) appeared in SoCG 2020. %
      \SICOMP{%
         \funding{Work by TMC was partially supported by NSF award
            CCF-1814026. Work by SHP and MJ was supported by AF awards
            CCF-1421231 and CCF-1907400.}  }%
   }%
}%

\NotSICOMP{%
   \date{\today}%

}

\ifx\GenSICOMPVer\undefined%
\else

\headers{Optimal Algorithms for Geometric Centers and
   Depth}{T. M. Chan, S. Har-Peled, and M. Jones}

\fi

\author{%
   Timothy M. Chan%
   \TimothyThanks{\NotSICOMP{Work was partially supported by NSF award
         CCF-1814026.}}%
   \and%
   Sariel Har-Peled%
   \SarielThanks{\NotSICOMP{Work on this paper was partially supported
         by NSF AF award CCF-1421231 and CCF-1907400.}}%
   \and%
   Mitchell Jones%
   \MitchellThanks{}%
}%

\begin{document}

\maketitle

\begin{abstract}
    We develop a general randomized technique for solving implicit
    linear programming problems, where the collection of constraints
    are defined implicitly by an underlying ground set of elements. In
    many cases, the structure of the implicitly defined constraints
    can be used to obtain faster linear program solvers.

    We apply this technique to obtain near-optimal algorithms for a
    variety of fundamental problems in geometry. For a given point set
    $P$ of size $n$ in $\Re^d$, we develop algorithms for computing
    geometric centers of a point set, including the centerpoint and
    the Tukey median, and several other more involved measures of
    centrality.  For $d=2$, the new algorithms run in $O(n\log n)$
    expected time, which is optimal, and for higher constant $d>2$,
    the expected time bound is within one logarithmic factor of
    $O(n^{d-1})$, which is also likely near optimal for some of the
    problems.
\end{abstract}

\section{Introduction}

\paragraph*{Parametric search}

In the 1980s, Nimrod Megiddo came up with an ingenious technique for
solving efficiently many geometric optimization problems.  This
para\-metric-search technique \cite{as-eago-98, m-apcadsa-83} works by
parallelizing a decider procedure for the problem (i.e., an algorithm
that can solve the decision problem associated with the optimization
problem), and conceptually running it on the unknown parameter being
the optimal value. One then simulates the execution of this parallel
algorithm.  This reduces to resolving a large batch of parallel
comparisons performed by the algorithm (i.e., the critical values of
the problem), which is done by performing a binary search over these
critical values using a sequential decider algorithm. The details of
the resulting algorithm, being a mixed simulation of a parallel
algorithm, tends to be convoluted, complicated and
counter-intuitive. Nevertheless, this technique provides the optimal
or fastest deterministic algorithm for many geometric optimization
problems.

\paragraph*{Linear programming (\LP)} %
Remarkably, in roughly the same time, Nimrod Megi\-ddo
\cite{m-lpltw-84} came up with a linear time algorithm for linear
programming in constant dimension. His algorithm shares some ideas
with his parametric search technique. This algorithm can be
dramatically simplified (and in practice sped up) by using
randomization \cite{s-sdlpchme-91, k-srsa-92, c-lvalipds-95,
   msw-sblp-96}.  Here is a quick sketch of Seidel's algorithm
\cite{s-sdlpchme-91}---it randomly permutes the constraints, inserts
them one by one, and checks whether an inserted constraint violates
the current optimal solution. If so, it recurses on the offending
constraint (and the prefix of the constraints inserted so far).  The
probability that the $i$\th constraint violates the current solution
is $O(1/i)$, which implies that the expected number of recursive calls
in the top level is $O( \log n)$.  Since the violation check takes
constant time, and the recursion depth is bounded, this readily
implies a running time that is near linear. A somewhat more careful
analysis shows that the expected running time is linear.

\begin{figure}
    \centering
    \begin{tabular}{cc}
      \includegraphics[page=1,width=0.39\linewidth]{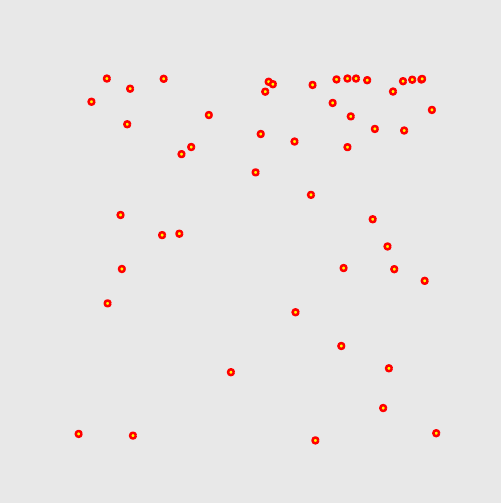}
      &
        \includegraphics[page=2,width=0.39\linewidth]{figs/yolk_2}
      \\
      (A) & (B)\\
      \includegraphics[page=3,width=0.39\linewidth]{figs/yolk_2}
      &
        \includegraphics[page=4,width=0.39\linewidth]{figs/yolk_2}\\
      (C) & (D)\\
    \end{tabular}
    \caption{%
       (A) Points. %
       (B) Median lines and the extremal yolk.  (C) All lines and the
       egg.  (D) Points with the extremal yolk and the egg.}
    \figlab{example}
\end{figure}%

\paragraph*{Randomization for parametric search} %
It is by now well known \cite{ov-psmp-04} that randomization can be
used as a replacement to parametric search, resulting in simpler (and
in many cases faster) algorithms. In particular, Chan
\cite{c-garot-99} identified a surprisingly simple and efficient
algorithmic technique that can be used to solve many of these
geometric optimization problems (it is similar in spirit to Seidel's
algorithm for \LP).  Specifically, imagine a maximization problem
where one has a fast decider algorithm that can tell us whether a
given value is larger or smaller than the optimal value of the given
instance.  Furthermore, assume that the problem at hand can be
(quickly) divided into a small number (e.g., constant) of smaller
instances, such that the value of one of these instances is the
desired optimal value, and all the other instances have values that
are not larger. Chan's algorithm now randomly orders these
subproblems, and solves the problem recursively on each subproblem,
but only if the (fast) decider indicates that the current subproblem
contains a higher value solution than the one found so far. If there
are $t$ subproblems, the algorithm in expectation performs only
$O( \log t)$ recursive calls. The result is a significantly simpler
randomized algorithm (which uses the decision algorithm only as a
black box), which is also faster and has none of the
logarithmic-factor slowdowns that parametric-search suffers from.

In this paper, we develop a generalization of the above randomized
optimization technique. This generalized technique is interesting in
its own right, as it can handle certain linear programming (\LP{})
problems, where the constraints are too numerous to be explicitly
stated and are thus specified implicitly. We apply this technique to a
variety of problems, discussed next.

\subsection{Motivation \& problems studied}
\seclab{motivation}

\subsubsection{Tukey depth}%

\begin{defn}
    \deflab{tukey:depth}%
    Given a set $\PS$ of $n$ points in $\Re^d$, the \emphi{Tukey
       depth} of a point $\pa \in \Re^d$ is
    \begin{equation*}
        \min_{\hs^+:\, \text{halfspace containing } \pa} |\PS\cap \hs^+|.
    \end{equation*}
\end{defn}

The task at hand is to compute a \emphi{Tukey median}, that is, a
point $\pa \in \Re^d$ with maximum Tukey depth. By the centerpoint
theorem, there is always a point of Tukey depth $\geq n/(d+1)$ in
$\Re^d$.

Notions of depths for point data sets are important in statistical
analysis.  The above definition (also called {\em location depth\/},
{\em data depth\/}, and {\em halfspace depth\/}) is among the most
well-known and was popularized by John Tukey~\cite{t-mpd-75}, who
suggested using the corresponding depth {\em contours\/} (boundaries
of regions of all points with equal depth) to visualize data.  A Tukey
median can serve as a point estimator for the data set (a ``center'')
which is robust against outliers, does not rely on distances, and is
invariant under affine transformations~\cite{rr-cbtm-98, rr-cdcbpc-96,
   s-smd-90}.

Because of the applications to statistics, the issue of designing
efficient algorithms to find Tukey medians and their relatives---for
example, a point with maximum \emph{Liu/simplicial depth}, minimum
\emph{Oja depth}, or maximum \emph{convex-layers/peeling depth}, and a
line or flat with maximum {\em regression depth\/}---has attracted a
great deal of attention from researchers in computational geometry.
See
\cite{acgst-lbcsd-02,alst-abmftfl-03,be-mrd-02,gsw-gm-92,vkmrsss-eamrd-08,
   ls-cmdpp-00,ls-chdp-03,ls-oa-03,mrrssss-ecldcmcg-03} for the
definitions of these concepts and relevant algorithms.

\typeout{1.1.2. Voting games and the yolk}

\subsubsection{Voting games and the yolk} %
Suppose there is a collection of $n$ voters in $\R^d$, where each
coordinate represents a specific ideology. In each coordinate, each
voter has a value representing their stance on a given ideology. One
can interpret $\R^d$ as a \emph{policy space}, and each point in
$\R^d$ represents a single policy. In the Euclidean spatial model, a
voter $\pa \in \R^d$ always prefers policies which are closer to $\pa$
under the Euclidean norm. For two policies $x, y \in \R^d$ and a set
of voters $\PS \subset \R^d$, $x$ \emph{beats} $y$ if more voters in
$\PS$ prefer policy $x$ compared to $y$. A plurality point is a policy
which beats all other policies in $\R^d$. For $d = 1$, the plurality
point is the median voter (when $n$ is odd) \cite{b-orgdm-48}. However
for $d > 1$, a plurality point is not always guaranteed to exist
\cite{r-nndshemr-79}. It is known that one can test whether a
plurality point exists (and if so, compute it) in $O(dn\log n)$ time
\cite{dbgm-facpp-18}. Note that the plurality point is a point of
Tukey depth $\ceil{n/2}$---in general this is the largest possible
Tukey depth any point can have; while the centerpoint is a point that
guarantees a ``respectable'' minority of size at least $n/(d+1)$.

Since plurality points may not always exist, one generalization of a
plurality point is the yolk \cite{m-cdifpsc-86}. A hyperplane is a
\emphi{median hyperplane} if the number of voters lying in each of the
two closed halfspaces (bounded by the hyperplane) is at least
$\ceil{n/2}$. The \emphi{yolk} is the ball of smallest radius
intersecting all such median hyperplanes. Note that when a plurality
point exists, the yolk has radius zero (equivalently, all median
hyperplanes intersect at a common point).

In terms of real world politics, one can think of the yolk as
representing an area of ambiguity where the policy of a political
party might be located. Such an ambiguity might be intentional, or the
natural consequence of forming a party made out of people with
differing views.

We also consider the following restricted problem.  A hyperplane is
\emphi{extremal} if and only if it passes through $d$ input points,
under the assumption that the points are in general position. The
\emphi{extremal yolk} is the ball of smallest radius intersecting all
extremal median hyperplanes. Importantly, the yolk and the extremal
yolk are different problems---the radius of the yolk and extremal yolk
can differ \cite{st-lml-92}.

\subsubsection{The egg of a point set} %
\seclab{egg:problem}%
A problem related to computing the yolk is the following: For a set of
$n$ points $\PS$ in $\R^d$, compute the smallest radius ball
intersecting all extremal hyperplanes of $\PS$ (i.e., all hyperplanes
passing through $d$ points of $\PS$). Such a ball is the \emphi{egg}
of $\PS$. See \figref{example} for an illustration of the yolk and egg
of a point set.

\subsubsection{Linear programs with many implicit constraints} %
Many of the problems discussed above (e.g., computing the Tukey median
or egg of a point set) can be written as an \LP{} with $\Theta(n^d)$
constraints, defined implicitly by the point set $\PS$. One can apply
Seidel's algorithm \cite{s-sdlpchme-91} (or any other linear time \LP
solver in constant dimension) to obtain an $O(n^d)$ expected time
algorithm for our problems.  However, as each $d$-tuple of points
forms a constraint, it is natural to ask if one can obtain a faster
algorithm in this setting. Specifically, we are interested in the
following problem: Let $I$ be an instance of a $d$-dimensional \LP
specified via a set of $n$ entities $\PS$, where each $k$-tuple of
$\PS$ induces a linear constraint in $I$, for some (constant) integer
$k$. The problem is to efficiently solve $I$, assuming access to some
additional subroutines. (We would also be interested in the more
general settings where not all the tuples induce constraints.)

\subsection{Previous work}

\subsubsection{On computing a Tukey median}
For dimension $d=1$, a Tukey median corresponds to the standard
median, and can be computed in linear time~\cite{clrs-ia-01}, with the
maximum depth being (exactly) $\ceil{n/2}$.

For $d >1$, the maximum Tukey depth is between $\ceil{n/(d+1)}$ and
$\ceil{n/2}$. The lower bound follows from the existence of a
\emphi{centerpoint}, which follows from Helly's theorem. A centerpoint
has depth at least $\ceil{n/(d+1)}$.  The first nontrivial algorithmic
result in the plane, by Cole \etal \cite{csy-khrp-87}, presented an
algorithm for computing a centerpoint in $O(n\log^5n)$ time, using a
two-level application of {parametric search} \cite{m-apcadsa-83}.
Cole's refined parametric-search technique \cite{c-sdsnofsa-87}
subsequently reduced the time bound to $O(n\log^3n)$.  Later, an
$O(n)$ time algorithm for centerpoints in the plane was discovered by
Jadhav and Mukhopadhyay~\cite{jm-ccfpsplt-94}, using a clever
prune-and-search approach.  This algorithm does not solve the (more
general) Tukey median problem.

In 1991, \Matousek \cite{m-ccpps-90} described an algorithm that
decides, in $O(n\log^4n)$ time, whether the maximum Tukey depth is at
least a given value~$k$, using also a two-level parametric search as a
subroutine.  His algorithm constructs a description of the entire
region of all points at depth at least $k$.  Consequently, a binary
search over $k$ yields the maximum Tukey depth and a Tukey median in
$O(n\log^5n)$ time.

In 2000, Langerman and Steiger~\cite{ls-cmdpp-00} obtained a faster
decision algorithm with an $O(n\log^3n)$ running time by using an
alternative to parametric search. This algorithm avoids constructing
the entire depth region.  The additional binary search then leads to
an $O(n\log^4n)$ time bound for Tukey median.  Subsequently, Langerman
and Steiger~\cite{ls-oa-03} showed that the Tukey median problem
itself can be solved in $O(n\log^3n)$ time.  Some extra logarithmic
factors seem inherent in their binary-search-like approach.  There is
an $\Omega(n\log n)$ lower bound on the time complexity of %
(i) computing the maximum Tukey depth, %
(ii) deciding whether the maximum depth is at least~$k$, %
(iii) the depth value of just a single point~$q$, %
or (iv) finding a Tukey median that is extreme along a given direction
\cite{alst-abmftfl-03, ls-cmdpp-00}.  We conjecture that the
$\Omega(n\log n)$ lower bound holds for finding an arbitrary Tukey
median as well.

Extensions to $d=3$ were also considered.  An $O(n^2\,\polylog\, n)$
algorithm for computing a 3-dimensional centerpoint was given by Naor
and Sharir~\cite{ns-cpcptd-90}.  More recently, Agarwal, Sharir, and
Welzl~\cite{AgarwalSW08} and Oh and Ahn~\cite{oa-ccr-19} gave more
near-quadratic algorithms (with extra $n^{o(1)}$ or polylogarithmic
factors) for computing the entire region of all points of depth at
least $k$ in 3 dimensions.

Note that the problem is difficult because of our insistence on using
exact depth values.  Approximate versions of the problem can be solved
considerably more quickly; for example,
see~\cite{cemst-acpir-96,hj-jcps-19,m-ccpps-90}.

\subsubsection{On computing the yolk} %
Both the yolk and the extremal yolk have been studied in the
literature.  The first polynomial time exact algorithm for computing
the yolk in $\R^d$ was by Tovey in $O\bigl(n^{(d+1)^2}\bigr)\Bigr.$
time---in the plane, the running time can be improved to $O(n^4)$
\cite{t-ptacyfd-92}.  Following Tovey, recent results have focused on
computing the yolk in the plane. In 2018, de Berg \etal
\cite{dbgm-facpp-18} gave an $O(n^{4/3} \log^{1 + o(1)} n)$ time
algorithm\footnote{ Actually the running time can be slightly improved
   to $O(n^{4/3})$, using a known randomized algorithm for median
   levels in the plane~\cite{Chan99}.  }  for computing the yolk. The
running time follows from the best known upper bound on the number of
combinatorially distinct median lines, which is $O(n^{4/3})$
\cite{d-ibpksrp-98}.  Obtaining a faster exact algorithm remained an
open problem.  Gudmundsson and Wong \cite{gw-cysvg-19,
   gw-cysvg-19-arxiv} presented a $(1+\eps)$-approximation algorithm
with $O(n \log^7 n \log^4 \eps^{-1} )$ running time. An unpublished
result of d{}e Berg \etal \cite{bcg-cysvg-19} achieves a randomized
$(1+\eps)$-approximation algorithm for the extremal yolk running in
expected time $O(n \eps^{-3} \log^3 n)$.

\subsubsection{On computing the egg} %
The egg of a point set in $\R^d$ can be computed by solving a linear
program with $\Theta(n^d)$ constraints. The egg is a natural extension
to computing the yolk, and thus obtaining faster exact algorithms is
of interest. The authors are not aware of any previous work on this
specific problem. Bhattacharya \etal \cite{bjmr-oasirp-94} gave an
algorithm which computes the smallest radius ball intersecting a set
of $m$ hyperplanes in $O(m)$ time, when $d = O(1)$, by formulating the
problem as an \LP{} (see also \lemref{egg:lp:type}). However we
emphasize that in our problem the set of hyperplanes is implicitly
defined by the point set $\PS$, and is of size $\Theta(n^d)$ in
$\Re^d$.

\begin{table}[t]
    \begin{center}
        \scriptsize
        \begin{tabular}{*{4}{|l}|}
          \multicolumn{1}{c|}{$d=2$}
          & \multicolumn{1}{c|}{$(1+\eps)$-approx}
          & \multicolumn{1}{c|}{Exact}
          & \multicolumn{1}{c}{Our results (Exact)}\\
          \hline\hline
          Extremal yolk
          & \ColY{$O(n \eps^{-3}\log^3 n)$}{\cite{bcg-cysvg-19}}
          & \ColY{$O(n^{4/3} \log^{1 + \epsilon} n)$}{\cite{dbgm-facpp-18}}
          & \ColY{$O(n \log n)$ }{\thmref{discrete:yolk}} \\
          \hline
          Yolk
          & \ColY{$O(n\log^7 n \log^4 \eps^{-1})$}{\cite{gw-cysvg-19-arxiv}}%
          & \ColZ{$O(n^{4/3} \log^{1 + \epsilon} n)$}%
            {Variant of \cite{dbgm-facpp-18}}{3.5cm}%
          & \ColY{$O(n \log n)$ }{\thmref{cont:yolk}} %
          \\
          \hline
          \multicolumn{1}{c|}{$d\ge 3$}
          & \multicolumn{3}{c}{}
          \\%
          \hline
          Extremal yolk
          & \multicolumn{1}{c|}{?}
          & \ColZ{$O(n^d)$}{Known techniques}{3cm}%
          &\ColY{$O(n^{d-1} \log n)$ }{\thmref{discrete:yolk}}%
          \\
          \hline%
          Yolk
          & \multicolumn{1}{c|}{?}
          & \ColZ{$O(n^d)$}{Known techniques}{3cm}%
          &\ColY{$O(n^{d-1})$ } {\thmref{cont:yolk}}%
          \\
          \hline%
        \end{tabular}%
    \end{center}
    \caption{Some previous work on the yolk and our results. Existing
       algorithms are deterministic, while the running time of our
       algorithms holds in expectation.}
    \tbllab{results}
\end{table}

\subsection{Our results}

In this paper we develop a generalized technique for solving \LP{}s
with many implicitly defined constraints. The new technique has
specific requirements to be met so it can be used, and these are
spelled out in \secref{settings:details}. Informally, these
requirements are:
\begin{compactenumI} %
    \smallskip%
    \item The problem can be solved in constant time for constant size
    instances.
    
    \smallskip%
    \item Given a candidate solution for an instance of size $n$, one
    can verify that it is optimal, in $\DX{n}$ time.
    
    \smallskip%
    \item One can break the given instance, in $\DX{n}$ time, into a
    constant number of smaller (by a constant factor) instances, such
    that the union of the implicit constraints they induce is the set
    of original constraints.

    \smallskip%
    \item The function $\DX{n}$ grows fast enough.
\end{compactenumI}
\medskip
Under these requirements, the implicit \LP problem can be solved in
$O\bigl( \DX{n}\bigr)$ time.

In \secref{many:cns} we state the technique and prove the key result
(\thmref{l:p:implicit}).  The technique builds on the work of Chan
\cite{c-garot-99} and leads to efficient algorithms for the following
problems. Throughout, let $\PS \subset \R^d$ be a set of $n$ points in
general position:

\begin{compactenumA}[leftmargin=0.8cm]
    \item In \secref{tukey:depth} we show that the point of maximum
    Tukey depth can be computed in $O(n^{d-1} + n\log n)$ expected
    time. As the problem of detecting affine degeneracies (the
    existence of $d$ points on a common hyperplane) among $n$ points
    in $\R^{d-1}$ is believed to require $\Omega(n^{d-1})$
    time~\cite{e-nlbchpod-99} and can be reduced to computing the
    maximum Tukey depth in $\R^d$, our result is likely to be optimal
    for $d \ge 3$ as well.  Note that as a byproduct, we get an
    improved $O(n^2)$ time randomized algorithm for centerpoints in
    $\R^3$. %

    \medskip
    \item In \secref{depth:region} we show how to compute in the plane
    the convex polygon forming the points of Tukey depth at least
    $k$. The new algorithm has running time $O(n \log^2 n)$, and
    improves over the work of \Matousek~\cite{m-ccpps-90}, that worked
    in $O(n \log^4 n)$ time.
    
    \medskip%
    \item The yolk of $\PS$ can be computed \emph{exactly} in
    $O(n^{d-1} + n \log n)$ expected time.  The extremal yolk can be
    computed in $O(n^{d-1} \log n)$ time.  Hence in the plane, the
    yolk can be computed in $O(n\log n)$ expected time.  This improves
    all existing algorithms (both exact and approximate)
    \cite{t-ptacyfd-92, dbgm-facpp-18, gw-cysvg-19-arxiv, gw-cysvg-19,
       bcg-cysvg-19} for computing the yolk in the plane, and our
    algorithm easily generalizes to higher dimensions.  See
    \tblref{results} for a summary of our results and previous work.

    \medskip%
    \item By a straightforward modification of the above algorithm, in
    \lemref{higher:dim:egg}, we prove that the egg of $\PS$ can be
    computed in $O(n^{d-1} \log n)$ expected time.

    \medskip%
    \item %
    Let $\PlanesX{k}(\PS)$ be the collection of all open halfspaces
    which contain more than $n-k$ points of $\PS$. Consider the convex
    polygon $\TP_k = \bigcap_{\hs \in \PlanesX{k}(\PS)}^{} \hs$.  Observe
    that $\TP_1$ is the convex hull of $\PS$, with
    $\TP_1 \supseteq \TP_2 \supseteq \cdots$. The centerpoint theorem
    implies that $\TP_{n/(d+1)}$ is non-empty (and contains the
    centerpoint). The Tukey depth of a point $\pq$ is the minimal $k$
    such that $\pq \in \TP_k \setminus \TP_{k+1}$.

    When $\TP_k$ is non-empty, the \emphi{center ball} of $\PS$ is the
    ball of largest radius contained inside $\TP_k$.  For $\TP_k$
    empty, we define the \emphi{Tukey ball} of $\PS$ as the smallest
    radius ball intersecting all halfspaces of $\PlanesX{k}(\PS)$.

    In \secref{tukey:center} we show that the Tukey ball and center
    ball can both be computed in
    $\tldO\bigl( k^{d-1}\big[1 + (n/k)^{\floor{d/2}}\big] \bigr)$
    expected time (see \lemref{tukey:ball:k} and
    \lemref{center:ball:k}, respectively), where $\tldO$ hides polylog
    terms.  In particular, when $k$ is a (small) constant, a point of
    Tukey depth $k$ can be computed in time
    $\tldO(n^{\floor{d/2}})$. As mentioned above, for
    $k \leq n/(d+1)$, the centerpoint has Tukey depth $\geq k$. As
    such, the issue here is computing such a point (and not deciding
    its existence).  This improves on the algorithm for computing a
    point of Tukey depth at least $k$, when $k \ll n$.
    \medskip%
    \item In \secref{more:apps}, we present the last application:
    Given a set $L$ of $n$ lines in the plane, the \emphi{crossing
       distance} between two points $\pa, \pq \in \R^2$ is the number
    of lines of $L$ intersecting the segment $\pa\pq$.  Given a point
    $\pq \in \R^2$ not lying on any lines of $L$, the disk of smallest
    radius containing all vertices of $\ArrX{L}$ within crossing
    distance at most $k$ from $\pq$ can be computed in $O(n\log n)$
    expected time.  See \lemref{cross:verts}.

\end{compactenumA}

\paragraph*{Paper organization} %
We provide some needed preliminaries in \secref{prelim}.  We study
some variants of \LP-type problems in \secref{variants} --
specifically, ranking and batched \LP-type problems. The main
technique is presented in \secref{many:cns}.  We present an algorithm
for the Tukey depth problem in \secref{tukey:depth}.  The algorithm
for the extremal yolk is presented in \secref{extremal:yolk}.  The
algorithm for the continuous yolk is presented in
\secref{continuous:yolk}.  The algorithms for Tukey ball and center
ball are presented in \secref{tukey:center}.  In \secref{more:apps},
we present the algorithm for computing the smallest disk within
certain crossing distance.  We conclude in \secref{conclusions} with a
few final remarks.

\section{Preliminaries}
\seclab{prelim}

\paragraph*{Notation} %
Throughout, the $O$ notation hides factors which depend (usually
exponentially) on the dimension $d$. Additionally, the $\tldO$
notation hides factors of the form $\log^c n$, where $c$ is a constant
that may depend on $d$.

\subsection{Duality}

\begin{definition}[Duality]
    \deflab{duality} The \emphi{dual hyperplane} of a point
    $\pa = (\pa_1, \ldots, \pa_d) \in \R^d$ is the hyperplane
    $\Dual{\pa}$ defined by the equation
    $x_d = -p_d + \sum_{i=1}^{d-1} x_i p_i$. The \emphi{dual point} of
    a hyperplane $h$ defined by $x_d = a_d + \sum_{i=1}^{d-1} a_i x_i$
    is the point
    \begin{equation*}
        \Dual{h} = (a_1, a_2, \ldots, a_{d-1}, -a_d).        
    \end{equation*}
\end{definition}

\begin{fact}
    \factlab{reverse}%
    Let $\pa$ be a point and let $h$ be a hyperplane. Then $\pa$ lies
    above $h$ $\iff$ the hyperplane $\Dual{\pa}$ lies below the point
    $\Dual{h}$.
\end{fact}

Given a set of objects $T$ (e.g., points in $\R^d$), let
$\Dual{T} = \Set{\Dual{x}}{x \in T}$ denote the dual set of objects.

\subsection{\TPDF{$k$}{k}-Levels}

\begin{definition}[Levels]
    \deflab{levels}%
    \deflab{top:bottom:lvls}%
    For a collection of hyperplanes $H$ in $\R^d$, the \emphi{level}
    of a point $\pa \in \R^d$, denoted by $\levelX{\pa}$, is the
    number of hyperplanes of $H$ lying on or below $\pa$.  The bottom
    \emphi{$k$-level} of $H$ is the (closure of the) union of points
    in $\R^d$ which have level equal to $k$, and let $\BY{k}{H}$
    denote the set of all such points.  The (bottom)
    \emphi{$(\leq k)$-level} of $H$ is the union of points in $\R^d$
    which have level at most $k$.  Let
    $\BY{<k}{H} = \bigcup_{i=0}^{k-1} \BY{i}{H}$.

    The top $k$-level is defined analogously (i.e., all points that
    have $k$ hyperplanes above them).  We denote the \emph{top
       $k$-level} by $\TY{k}{H}$.  Let
    $\TY{<k}{H} = \bigcup_{i=0}^{k-1} \TY{i}{H}$.
\end{definition}

By \factref{reverse}, if $h$ is a hyperplane which contains $k$ points
of $\PS$ lying on or above it, then the dual point $\Dual{h}$ is a
member of the $k$-level of $\Dual{\PS}$.

\subsection{Zones of surfaces}

For a set of hyperplanes $H$, denote the arrangement of $H$ by
$\ArrX{H}$ (see, e.g., \cite{bcko-cgaa-08}).

\begin{definition}[Zone of a surface]
    For a collection of hyperplanes $H$ in $\R^d$, the complexity of a
    cell $\psi$ in the arrangement $\ArrX{H}$ is the number of faces
    (of all dimensions) which are contained in the closure of
    $\psi$. For a $(d-1)$-dimensional surface $\gamma$, the
    \emphi{zone} $\ZoneY{\gamma}{H}$ of $\gamma$ is the subset of
    cells of $\ArrX{H}$ which intersect $\gamma$. The
    \emphi{complexity of a zone} is the sum of the complexities of the
    cells in $\ZoneY{\gamma}{H}$.
\end{definition}

The complexity of a zone of a hyperplane is known to be
$\Theta(n^{d-1})$ \cite{ess-ztha-93}. For general algebraic surfaces
it is larger by a logarithmic factor. Furthermore, the cells in the
zone of a surface can be computed efficiently using lazy randomized
incremental construction \cite{dbds-lric-95}.

\begin{lemma}[\cite{aps-ozsha-93,dbds-lric-95}]
    \lemlab{zone:complexity} Let $H$ be a set of $n$ hyperplanes in
    $\R^d$ and let $\gamma$ be a $(d-1)$-dimensional algebraic surface
    of degree $\delta$.  The complexity of the zone
    $\ZoneY{\gamma}{H}$ is $O(n^{d-1} \log n)$, where the hidden
    constants depend on $d$ and $\delta$. The collection of cells in
    $\ZoneY{\gamma}{H}$ can be computed in $O(n^{d-1} \log n)$
    expected time.
\end{lemma}

\subsection{Cuttings}
\seclab{cuttings}

\begin{defn}[Cuttings]
    Given $n$ hyperplanes in $\Re^d$, a \emphi{$(1/c)$-cutting} is a
    collection of interior disjoint simplices covering $\Re^d$, such
    that each simplex intersects at most $n/c$ hyperplanes.
\end{defn}

\begin{lemma}[\cite{c-chdc-93}]
    \lemlab{cuttings}%
    Given a collection of $n$ hyperplanes in $\Re^d$, a
    $(1/c)$-cutting of size $O(c^d)$ can be constructed in
    $O(nc^{d-1})$ time.
\end{lemma}

\subsection{LP-type problems}
\seclab{tm:lptype}

An \LP{}-type problem, introduced by Sharir and Welzl
\cite{sw-cblprp-92}, is a generalization of a linear program. Let
$\CN$ be a set of constrains and $\Objv$ be an objective function. For
any $\CNA \subseteq \CN$, let $\ObjvX{\CNA}$ denote the value of the
optimal solution for the constraints of $\CNA$. The goal is to compute
$\ObjvX{\CN}$. If the problem is infeasible, let
$\ObjvX{\CN} = \infty$. Similarly, define $\ObjvX{\CN} = -\infty$ if
the problem is unbounded.

\begin{definition}
    \deflab{lp:type}%
    Let $\CN$ be a set of constraints, and let
    $\Objv : 2^\CN \to \Re \cup \{\infty, -\infty\}$ be an objective
    function. The tuple $(\CN, \Objv)$ forms an \emphi{\LP{}-type
       problem} if the following properties hold:
    \begin{compactenumA}
        \smallskip%
        \item \itemlab{monotonicity}%
        \textbf{Monotonicity.} For any
        $\CNa \subseteq \CNb \subseteq \CN$, we have
        $\ObjvX{\CNa} \leq \ObjvX{\CNb}$.

        \smallskip%
        \item \itemlab{locality}%
        \textbf{Locality.} For any $\CNa \subseteq \CNb \subseteq \CN$
        with $\ObjvX{\CNb} = \ObjvX{\CNa} > -\infty$, and for all
        $s \in \CN$,
        $\ObjvX{\CNb} < \ObjvX{\CNb + s} \iff \ObjvX{\CNa} <
        \ObjvX{\CNa + s}$, where $\CNa + s = \CNa \cup \brc{s}$.
    \end{compactenumA}
\end{definition}

A \emphi{basis} of a set $\CN' \subseteq \CN$ is an inclusion-wise
minimal subset $\Basis \subseteq \CN'$ with
$\ObjvX{\Basis} = \ObjvX{\CN'}$.  The \emphi{combinatorial dimension}
$\CD$ is the maximum size of any feasible basis of any subset $\CN'$
of $\CN$. Throughout, we consider $\CD$ to be a constant. For a basis
$\Basis \subseteq \CN$, a constraint $\cons \in \CN$ \emphi{violates}
the current solution induced by $\Basis$ if
$\ObjvX{\Basis + \cons} > \ObjvX{\Basis}$.  \LP{}-type problems with
$n$ constraints can be solved in randomized time $O(n)$, hiding
constants depending (exponentially) on $\CD$ \cite{c-lvalipds-95},
where the bound on the running time holds with high probability.

\begin{remark}
    The aforementioned algorithms for solving \LP-type problems
    require certain primitives to be provided, such as testing for a
    basis violation, and computing the basis for a small set of
    constraints. In the following, we assume that such primitives are
    provided when considering any \LP-type problem.
\end{remark}

\section{Variants of \LP-type problems}
\seclab{variants}

\subsection{Ranking \LP-type problem}
\seclab{ranking:LP}

Let $\CN$ be a set of constraints, and assume that each constraint
$\hp $ has an associated rank $\rankX{\hp} \in \Re^+$ (importantly for
our purposes, the ranks are not necessarily distinct).  Assume that
$(\CN, \Objv)$ forms an instance of \LP-type (minimization)
problem. This instance might not be feasible (i.e.,
$\ObjvX{\CN} = +\infty$), so consider the parameterized instance.
Here, for a number $\alpha \in \Re$, we consider the \LP-type problem
instance $\Instance(\alpha)$ formed by the set of constraints
\begin{equation*}
    \CN_{\geq \alpha} = \Set{ \hp \in \CN}{\rankX{\hp} \geq \alpha}.
\end{equation*}
Let $\rankMinX{\CN}$ be the minimum positive real value of $\alpha$
such that $\CN_{\geq \alpha}$ is feasible.  For a set of constraints
$\CNa \subseteq \CN$, consider the target function
$\ObjvB(\CNa) = (\alpha,\beta )$, where $\alpha=\rankMinX{\CNa}$ and
$\beta = \ObjvX{\CNa_{\geq \alpha}}$. Let $\prec$ be the
lexicographical ordering of $\Re^2$ (i.e., $(x,y)\prec (x',y')$ $\iff$
$x < x'$ or [$x=x'$ and $y < y'$]).  The new optimization problem
$(\CN, \ObjvB )$, is to compute $\ObjvB( \CN)$.  Thus computing the
minimum $\alpha$, such that $\CN_{\geq \alpha}$ is feasible, and the
associated original \LP value for this subset.  We refer to
$(\CN, \ObjvB)$ as the \emphi{ranking} problem associated with
$(\CN, \Objv)$.

\begin{lemma}
    \lemlab{ranking:lp}%
    Given an \LP-type problem $(\CN, \Objv)$ of combinatorial
    dimension $d$, with associated ranks on the constraints of $\CN$,
    the ranking problem $(\CN, \ObjvB )$ is an \LP-type problem of
    combinatorial dimension $d+1$.
\end{lemma}
\begin{proof}
    The proof is straightforward, and we include it only for the sake
    of completeness.

    The basis of $\ObjvB( \CN)$ is a minimal set
    $\Basis \subseteq \CN$, such that
    $\ObjvB(\Basis) = \ObjvB( \CN )$.  A constraint $\cnx \in \CN$
    \emph{violates} $\Basis$, if $\rankX{\cnx} \geq \rankX{\Basis}$
    and $\Objv(\Basis + \cnx) > \Objv(\Basis)$, where
    $\rankX{\Basis} = \min_{\hp \in \Basis} \rankX{ \hp}$. As such, a
    basis of the ranking problem, is a basis of the original problem
    with potentially one additional constraint that realizes/records
    the minimum realizable rank subset.

    We now verify the required \LP properties from \defref{lp:type}:
    \begin{compactenumA}
        \smallskip%
        \item \textsc{Monotonicity}: For any
        $\CNa \subseteq \CNb \subseteq \CN$, let $\alpha$ be the
        minimum value such that $\CNb_{\geq \alpha}$ is feasible for
        $\Objv$. Clearly,
        $\CNa_{\geq \alpha} \subseteq \CNb_{\geq \alpha}$, which
        implies that $\CNa_{\geq \alpha}$ is feasible for
        $\Objv$. That is, we have $\ObjvB(\CNa) \preceq \ObjvB(\CNb)$.

        \medskip%
        \item \textsc{Locality}: Consider any
        $\CNa \subseteq \CNb \subseteq \CN$ with
        $\ObjvB(\CNb) = \ObjvB(\CNa) \succ (0, -\infty)$.

        Consider any $s \in \CN$, such that
        $(\alpha,\beta) = \ObjvB(\CNb) \prec \ObjvB(\CNb + s) =
        (\alpha',\beta')$. If $\alpha = \alpha'$ then
        $\beta = \Objv(\CNa_{\geq \alpha}) < \Objv(\CNa_{\geq \alpha}
        + s) = \beta'$.  The locality of $\Objv$ implies that
        $\Objv(\CNb_{\geq \alpha}) < \Objv(\CNb_{\geq \alpha} +
        s)$. This implies that $\ObjvB(\CNa) \prec \ObjvB(\CNa + s)$,
        which implies the locality property.
        
        Otherwise, $\alpha' > \alpha$. This implies that
        $\CNc_{\geq \alpha} + s$ is not feasible, which implies
        $\Objv(\CNb_{\geq \alpha}) < \Objv(\CNb_{\geq \alpha} + s)$.
        By the locality of $\Objv$, we have that
        $\Objv( \CNa_{\geq \alpha} + s ) > \Objv( \CNa_{\geq \alpha}
        )$. But that readily implies that
        $\ObjvB( \CNa + s ) \succ \ObjvB( \CNa ) = (\alpha,\beta)$.

        As for the other direction, if
        $\ObjvB(\CNa) \prec \ObjvB(\CNa + s) $ then by monotonicity,
        we have that $\ObjvB(\CNb) \prec \ObjvB(\CNb + s) $, which
        implies locality.
    \end{compactenumA}
\end{proof}

\subsection{Batched \LP}
\seclab{batched:lp}

Intuitively, one can group constraints of an \LP-type problem so that
each group form its own ``constraint'', and the modified problem
remains an \LP-type problem. This is captured by the following
definition.

\begin{definition}[Batched \LP{}-type problems]
    \deflab{batched}%
    Let $(\CN, \Objv)$ be an \LP{}-type problem.  A \emphi{batched}
    \LP-type problem is defined by the constraint set $2^\CN$ and the
    objective function
    $\ObjvExt : 2^{2^\CN} \to \R \cup \{\infty, -\infty\}$. For
    non-empty $\CNa_1, \ldots, \CNa_m \subseteq \CN$, define
    \begin{math}
        \ObjvExt(\{\CNa_1, \ldots, \CNa_m\}) := \Objv(\CNa_1 \cup
        \ldots \cup \CNa_m)
    \end{math}.
\end{definition}

\begin{lemma}
    \lemlab{same:cd}%
    Let $(\CN, \Objv)$ be an \LP{}-type problem of combinatorial
    dimension $\lpdim$, and let $\CNP = 2^\CN$. Then
    $(\CNP, \ObjvExt )$ is an \LP{}-type problem with combinatorial
    dimension $\lpdim$.
\end{lemma}
\begin{proof}
    The proof of this lemma is straightforward -- we provide a proof
    for the sake of completeness, but the reader is encouraged to skip
    it.  For any sets $\CNA \subseteq \CNB \subseteq \CNP$, we have
    that
    \begin{equation*}
        \ObjvExt(\CNA)%
        =%
        \Objv( \cup \CNA)%
        \leq%
        \Objv( \cup \CNB)%
        \leq
        \ObjvExt(\CNB),
    \end{equation*}
    by the monotonicity of $\Objv$, see \defref{lp:type}, where
    $\cup \CNA = \cup_{X \in \CNA} X$. This readily implies the
    monotonicity of $\ObjvExt$.

    For any $\CNA \subseteq \CNB \subseteq 2^\CN$ with
    $\ObjvExt( \CNB ) = \ObjvExt( \CNA ) > -\infty$, and for all
    $\CNs \in \CNP$, we have that
    \begin{align*}
      \ObjvExtX{\CNB} < \ObjvExtX{\CNB + \CNs}
      \SICOMP{&}
                \iff
                \ObjvX{\cup \CNB } < \ObjvX{\cup\CNB \cup \CNs}
                \iff
                \ObjvX{\cup \CNA } < \ObjvX{\cup\CNA \cup \CNs}
                \SICOMP{\\&}
      \iff
      \ObjvExtX{\CNA} < \ObjvExtX{\CNA + \CNs},
    \end{align*}
    by the locality of $\Objv$. This implies the locality of
    $\ObjvExt$.

    As for the combinatorial dimension, consider any family of sets
    \begin{equation*}
        \CNA = \{ \CNa_1, \ldots, \CNa_m \} \subseteq \CN,        
    \end{equation*}
    and let $\CNX = \{\cnx_1, \ldots, \cnx_t\}$ be the basis of
    $\ObjvX{ \cup_i \CNa_i}$. By assumption, $t \leq \lpdim$.  Assume,
    for simplicity of exposition, that $\cnx_i \in \CNa_i$, for all
    $i$. We have that $\CB = \{ \CNa_1,\ldots, \CNa_t \}$ is a basis
    for $\CNA$, for $\ObjvExt$, as
    \begin{equation*}
        \ObjvExtX{\CNA}%
        =%
        \ObjvX{\cup_i\CNa_i}        
        =%
        \ObjvX{ \CNX}
        \leq
        \ObjvExtX{\CB}%
        \leq 
        \ObjvExtX{\CNA}%
        \quad\implies\quad%
        \ObjvExtX{\CB}%
        =\ObjvExtX{\CNA}.
    \end{equation*}
    That readily implies that the combinatorial dimension of
    $(\CNP, \ObjvExt )$ is $\lpdim$.
\end{proof}

\subsection{Solving \LP-Type problem for bundles using %
   Seidel's algorithm}
\seclab{lp:bundles}

Let $\Instance = (\CN, \Objv)$ be an instance of an \LP{}-type problem
with combinatorial dimension $\lpdim$.

\begin{defn}
    Given a set $\CN$ of constraints, a subset of
    $\Bundle \subseteq \CN$ is a \emphi{bundle}. A collection
    $\Bundle_1, \ldots, \Bundle_n$ of bundles is a \emphi{cover} of
    $\CN$ if $\bigcup_i \Bundle_i = \CN$.
\end{defn}

For a given cover $\Bundle_1, \ldots, \Bundle_n$, and a basis
$\Basis$, assume that one can compute the basis of
$\Basis \cup \Bundle_i$ in $\DT$ time, for any $i$, by calling a
procedure \compBasis. Furthermore, assume that one can also check if
any of the constraints of $\Bundle_i$ violates $\Basis$, in $\DT$
time, by calling a procedure \violate. It is natural to ask if one can
solve $\Instance$ quickly. This question makes sense if $\DT$ is
sublinear in the number of constraints in a bundle.

We are going to use a variant of Seidel's algorithm for solving
\LP-type problems, which also works for violator spaces
\cite{s-sdlpchme-91, h-gaa-11,h-sppss-16} (other algorithms for
solving \LP-type problems can be used here).  For the sake of
completeness, we next sketch this algorithm.  See also
\figref{LP:type:algorithm}.

\begin{figure}[t]
    \centerline{%
       \begin{algorithmEnv}
           \solveLPType{}$\pth{\Basis_0, \Bundle_1, \ldots, \Bundle_m }$ \+\+\\
           // $\Basis_0$: initial basis \\
           $\permut{\Bundle_1', \ldots, \Bundle_n'}$: random
           permutation of
           $\Bundle_1, \ldots, \Bundle_m$.\\
           \For $i=1$ to $m$ \Do\+\\
           \If \violate{}($\Bundle_i'$, $\Basis_{i-1}$) \Then\+\\
           $\BasisB_i \leftarrow \compBasis\pth{ \Basis_{i-1}, \Bundle_i'}$\\
           $\Basis_i \leftarrow \solveLPType\pth{\BasisB_i \cup
              \Basis_0, \Bundle_1', \ldots, \Bundle_i'}$
           \-\\
           \Else \+\\
           $\Basis_i \leftarrow \Basis_{i-1}$
           \-\-\\
           \Return $\Basis_n$
       \end{algorithmEnv}%
    }
    \caption{The algorithm for solving \LP-type problems.}
    \figlab{LP:type:algorithm}
\end{figure}

The input is an initial basis $\Basis_0$ (made out of at most $\lpdim$
constraints), and $n$ bundles of constraints
$\Bundle_1, \ldots, \Bundle_n$.  The algorithm picks uniformly at
random a permutation $\pi$ of $\IRX{n} = \{1,\ldots, n\}$.  In the
$i$\th iteration, the algorithm adds $\Bundle_i' = \Bundle_{\pi(i)}$
to the current set of constraints, maintaining the optimal solution to
$\CNX_i= \{ \Bundle_1', \ldots, \Bundle_i' \} \cup \Basis_0$.  To this
end, the algorithm maintains a basis $\Basis_i \subseteq \CNX_i$ of
the solution for the constraints of $\CNX_i$. Initially, the algorithm
sets $\Basis_0 = \Basis$.

In the beginning of the $i$\th iteration, the algorithm uses the
violation test (i.e., \violate) to decide if any of the constraints of
$\Bundle_i'$ violates $\Basis_{i-1}$. If there is no violation, the
algorithm sets $\Basis_i$ to $\Basis_{i-1}$, and continues to the next
iteration.  Otherwise, the algorithm computes the basis $\BasisB_i$ of
$\Basis_{i-1} \cup \Bundle_i'$ by calling \compBasis. The algorithm
then computes $\Basis_i$ by calling itself recursively with the
initial ``basis'' $\BasisB_i \cup \Basis_0$ and the bundles
$\Bundle_1', \ldots, \Bundle_i'$ (this second call is on \LP-type
instance with smaller combinatorial dimension).  At the end, the
algorithm returns $\Basis_m$ as the basis of the solution.

The correctness of this algorithm is immediate by interpreting this
problem as an instance of batched \LP. See \lemref{same:cd}. In
particular, in this variant of the algorithm, the depth of the
recursion is at most $\lpdim$ \cite{h-sppss-16}, and as such the
initial ``basis'' set is of size at most $\lpdim^2$ (i.e., constant).
We need the following well known result.

\begin{lemma}[\cite{sw-cblprp-92, h-gaa-11, h-sppss-16}]
    \lemlab{num:calls:new}%
    The input is an instance $\Instance = (\CN, \Objv)$ of an
    \LP{}-type problem with constant combinatorial dimension $\lpdim$.
    In addition, the input also includes a cover $\CN$ by $n$ bundles
    $\Bundle_1, \ldots, \Bundle_n$, and procedures \violate and
    \compBasis as described above, where each call takes $\DT$
    time. For this input, the above algorithm computes the optimal
    solution to $\Instance$ in $O( n \DT )$ time, and
    $O\pth{(\lpdim \log m)^\lpdim}$ calls to \compBasis.
\end{lemma}

\begin{remark}
    It is possible to further reduce the number of calls to
    \compBasis\ to $O(\lpdim^{O(\lpdim)}\log m)$ by using Clarkson's
    randomized \LP algorithm~\cite{c-lvalipds-95} instead of Seidel's,
    though such an improvement will not be needed in our applications.
\end{remark}

\section{An optimization technique for implicit \LP{}-type problems}
\seclab{many:cns}

The main challenge in implementing the algorithm of
\secref{lp:bundles} for bundles is that we need to provide the
procedures \violate and \compBasis. A natural approach, if we have a
way to break a bundle into subbundles, is to use recursion. For this
scheme to make sense, the number of constraints defined by a bundle
has to be defined implicitly, and be superlinear in the number of
entities defining a bundle.

\subsection{Settings and basic idea}
\seclab{settings}

\subsubsection{Instance of implicit \LP}
\seclab{settings:details}

Let $(\CN, \Objv)$ be an \LP{}-type problem of constant combinatorial
dimension $\lpdim$. Let $\sza$ be an integer constant and $\szb > 1$
be another constant.  For an input space $\Pi$, suppose that there is
a function $\Impl : \Pi \to 2^\CN$ which maps inputs to sets of
constraints.  Furthermore, assume that for any input $\PS \in \Pi$ of
size $n$, we have the following properties: \smallskip%
\begin{compactenumI}[leftmargin=1.4cm,label={(P\arabic*)}]%
    \item \itemlab{base}%
    When $n = O(1)$, a basis for $\Impl (\PS)$ can be computed in
    constant time.
    
    \smallskip%
    \item \itemlab{verify}%
    We are given a violation test that, for a basis $\Basis$, decides
    if $\Basis$ satisfies $\Impl(\PS)$ in $\DX{n}$ time.  (Let
    \violate be the name of this procedure.)

    \smallskip%
    \item \itemlab{partition} In $\DX{n}$ time, one can construct sets
    $\PS_1, \ldots, \PS_\sza \in \Pi$, each of size at most $n/\szb$,
    such that $\Impl(\PS) = \bigcup_{i=1}^\sza \Impl(\PS_i)$.

    \smallskip
    \item \itemlab{sub} The function $\DX{n}/n^\eps$ is monotonically
    increasing for some constant $\eps>0$.
\end{compactenumI}
\smallskip%
The above is an \emphi{instance} of the implicit \LP problem.

\begin{remark}
    \remlab{replacement}%
    Note that the above condition on $\DX{n}$ in particular implies
    that for any $1 \leq k \leq n$, we have
    $\DX{n/k}/(n/k)^\eps \le \DX{n}/n^\eps$, i.e.,
    $\DX{n/k} \le \DX{n}/k^\eps$.
\end{remark}

\subsubsection{An assumption}

An annoying technicality is that to bound the running time of the
resulting algorithm, we need a strong assumption on the parameters
used in the above instance. We tackle issue next, but the casual
reader can safely skip to \secref{egg:plane}.

Specifically, for fixed constants $\eps \in (0,1)$, and $c > 1$, the
required property is that
\begin{equation}
    \frac{c \log^\lpdim \sza}{\szb^\eps} <  1,
    \eqlab{weird}
\end{equation}
where $\sza$ is the number of sets in the partition, and $n/\szb$ is
the bound on the size on each set of the partition, see
\itemref{partition}.  If the given instance does not have this
property, then one can modify the instance, to get a new instance that
has this property, as testified by the following.

\begin{lemma}
    \lemlab{weird}%
    Consider a fixed constant $\eps \in (0,1)$, and some fixed
    constant $c > 1$.  Given an instance of implicit \LP, one can
    create a new instance such that \Eqref{weird} holds. The
    asymptotic running time of the new partition procedure is the
    same.
\end{lemma}
\begin{proof}
    The idea is to apply the decomposition recursively for, say, $i$
    levels. We get sets $\PS_1, \ldots, \PS_m \subseteq \PS$, with
    $m \leq \sza^i$, where each set is of size at most
    $n/\szb^i$. This yields a finer decomposition into a larger number
    of smaller sets, and we can use this decomposition in the above
    settings.
    Then
    \begin{equation*}
        c \frac{ \log^\lpdim \sza^i}{ \szb^{\eps i}}
        =%
        c \frac{ i^\lpdim \log^\lpdim \sza}{ \szb^{\eps i}}        
        < 1,
    \end{equation*}
    by choosing $i$ to be a sufficiently large constant (depending
    only on $\sza,\szb,\lpdim,\eps,c$), since $\szb > 1$.  Thus, using
    this decomposition with $\sza^i$ sets, and shrinking factor
    $\szb^i$, implies a new instance of the problem that satisfies the
    claim.
\end{proof}

\subsubsection{Example: Egg in the plane}
\seclab{egg:plane}

Consider the egg problem in the plane, as defined in
\secref{egg:problem}.  The input space $\Pi$ is a set of $n$ points
$\PS$ in the plane. For a pair of points $\pa, \pb \in \PS$, consider
the line $\Line$ with equation $\alpha x + \beta y + \gamma =0$ that
passes through these two points, where $\alpha^2 + \beta^2 =
1$. Consider a point $\pc = (x,y,r)$ in three dimensions. The point
$\pc$ encodes a disk of radius $r$ centered at $(x,y)$. This disk
intersects $\Line$ if and only if
\begin{equation*}
    -r \leq \alpha x + \beta y + \gamma \leq r.
\end{equation*}
This condition corresponds to two linear inequalities, and let
$\Impl\bigl( \{\pa,\pb\})$ denote the set of these two
inequalities. For a set $\QS \subseteq \PS$, the associated set of
linear constraints is
\begin{equation*}
    \Impl\bigl( \QS )%
    =%
    \bigcup_{\pa, \pb \in \QS, \pa \neq \pb} \Impl\bigl( \{\pa,\pb\}).
\end{equation*}
Thus, to compute the smallest disk intersecting all the lines induced
by $\PS$, we need to solve the three-dimensional \LP defined by the
constraints of $\Impl\bigl( \PS )$---computing the feasible point with
minimum $z$ coordinate. This \LP can be solved in quadratic time in a
straightforward fashion---here we are interested in solving this
problem more quickly.

\begin{figure}[t]
    \centering{{\includegraphics{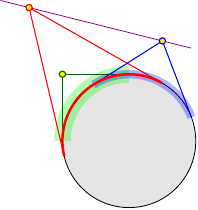}}}
    \caption{Finding an induced line avoiding the disk is equivalent
       to computing two arcs that intersect but do not contain each
       other in the associated circular arc graph.}
    \figlab{circular}%
\end{figure}

We next verify that the above settings apply. First, the problem can
be solved in constant time for a constant size point set. Next, given
a basis (i.e., a disk $\Disk$ in this case), and a set $\QS$ of
points, one can check (in near linear time) whether two points of
$\QS$ induce a line that avoids the disk $\Disk$.  An algorithm for
this problem is described, in somewhat more abstract settings, in
\secref{subproblem}. In this specific case, one can define a circular
arc graph on the boundary of $\Disk$ (see \figref{circular}), where
each point $\pa$ outside $\Disk$ induces an arc of all the points on
the disk boundary that see $\pa$. On this circular arc graph it is
enough to decide if there are two arcs that intersect but do not
contain each other, and this can be done, in $O(n \log n)$ time, using
``sweeping''.  As such, in this case, we have property
\itemref{verify} with $\DX{n} = O( n\log n)$.

The divide property, i.e., \itemref{partition}, is surprisingly
straightforward in this case. Partition the point set $\PS$ into $b$
helper sets $\QS_1, \ldots, \QS_{b}$ with
$\cardin{\QS_i} \geq \floor{ n/b}$, for
$i \in \IRX{b}= \{1,\ldots, b\}$. For any pair
$\{i,j\} \in\binom{\IRX{b}}{2}$ define the point set
$\PS_{i,j} = \QS_i \cup \QS_j$.  Thus, the constructed sets are the
members of
$\Set{\smash{\PS_{i,j}}}{\{i,j\} \in
   \smash{\binom{\IRX{b}}{2}}\bigr.}$.  As such, $\sza = \binom{b}{2}$
and $\szb = b/2$. Clearly, for any pair of points
$\{ \pa, \pb \} \in \binom{\PS}{2}$, there are indices $i,j$, such
that $\pa, \pb \in \PS_{i,j}$. This implies that
$\Impl( \PS ) = \Impl\bigl(\bigcup_{i,j} \PS_{i,j} \bigr)$.

The combinatorial dimension of the associated \LP is $3$.  We can
specifically choose $b=3$, but if we want to satisfy \Eqref{weird}
directly without invoking \lemref{weird}, we can pick the smallest $b$
such that $c\log^3\binom{b}{2}/(b/2)^\eps < 1$.

\subsection{The algorithm for solving implicit \LP}

\subsubsection{The algorithm in detail}

The basic idea is to modify the algorithm of \secref{lp:bundles} so
that \compBasis is implemented recursively, via refinement of the
current bundle. So, the input is a set of $n$ elements
$\PS \subseteq \Pi$, and the task is to solve the \LP-type problem
defined by the set of constraints $\Impl(\PS)$. For any set
$\PSA \subseteq \PS$, let $\Basis(\PSA)$ denote the basis of size at
most $\lpdim$ for the constraint set $\Impl(\PSA)$.  By requirement
\itemref{verify}, given a basis $\Basis$, we can decide if
$\Objv\bigl(\Basis \cup \Impl(\PSA)\bigr) > \Objv(\Basis)$ by calling
\violate.

\begin{figure}[h]
    \centerline{%
       \begin{algorithmEnv}
           \solveLPT{}$\pth{\Basis_0, \CNR_1, \ldots, \CNR_m }\Bigr.$
           \qquad// $\Basis_0$: initial basis \+\+
           \\
           $\permut{\CNR_1', \ldots, \CNR_n'}$: random permutation of
           $\CNR_1, \ldots, \CNR_m$.\\
           \For $i=1$ to $m$ \Do\+\\
           \If \violate{}($\CNR_i'$, $\Basis_{i-1}$) \Then\+\\
           $\BasisB_i \leftarrow \compBasis\pth{ \Basis_{i-1}, \CNR_i'}$\\
           $\Basis_i \leftarrow \solveLPT\pth{\BasisB_i \cup \Basis_0,
              \CNR_1', \ldots, \CNR_i'}$
           \-\\
           \Else \+\\
           $\Basis_i \leftarrow \Basis_{i-1}$
           \-\-\\
           \Return $\Basis_n$
       \end{algorithmEnv}%
    }
    \caption{The main procedure for solving the implicit \LP-type problems.}
    \figlab{LP:i:main}
\end{figure}

The main algorithm, $\solveLPT{}( \Basis, \CNR_1, \ldots, \CNR_m)$, is
given an initial basis $\Basis \subseteq \CN$, and $m$ subsets of
$\PS$. It returns the basis $\BasisB \subseteq \CN$ for the constraint
set $\Impl(\CNR_1) \cup \ldots \cup \Impl(\CNR_m) \cup \Basis$. This
algorithm implements \solveLPType{} modified to use the implicit
representation -- depicted in \figref{LP:i:main}.  To solve the
\LP{}-type problem of interest, invoke $\compBasis{}(\PS, \Basis)$,
where $\Basis$ is some initial basis.  Note, that \compBasis{} and
\solveLPType{} are mutually recursive calling each other in
turn. Importantly, the sizes of the subproblems decreases down the
recursion, implying that this algorithm indeed terminates.

\paragraph{Implementing \compBasis{}$( \PS, \Basis)$}

Here, $\PS$ is a set of $n$ entities, and $\Basis$ is an initial
basis.  The goal is to compute the basis for the set
$\Basis \cup \Impl(\PS)$. If $\PS$ is of constant size, we solve the
associated \LP{}-type problem in $O(1)$ time.

Otherwise, using the partition algorithm provided as part of the
instance (i.e., \itemref{partition}) compute sets
$\CNR_1, \ldots \CNR_\sza \subseteq \PS$, each of size at most
$n/\szb$. The problem is reduced to computing a basis for the set
$\Impl(\CNR_1) \cup \ldots \cup \Impl(\CNR_\sza)$.  Specifically, one
need to compute a basis for the set
$\{\Impl(\CNR_1), \ldots, \Impl(\CNR_\sza)\}$ of $\sza$ elements. By
\lemref{same:cd}, this new problem remains an \LP{}-type problem of
combinatorial dimension $\lpdim$. As such, we can invoke the
subroutine $\solveLPT{}( \Basis, \CNR_1, \ldots, \CNR_\sza)$ to solve
the extended \LP{}-type problem, and return the required basis.

\subsubsection{Analysis}

The procedure \compBasis{} is recursive, with the top most call being
of depth zero, which in turns calls \solveLPT{}. which in turn might
call \compBasis{} (these calls are of depth one), and so on.

\begin{lemma}
    \lemlab{levels}%
    A call to \compBasis{} of depth $i$ (might) result in a call to
    \solveLPT{}.  This call to \solveLPT{} triggers, in expectation,
    $O( \log^\lpdim \sza)$ calls to \compBasis{} of depth $i+1$, and
    $c_\lpdim \sza$ violation tests (at this level of the
    recursion---we are ignoring such calls performed in lower levels
    of the recursion).
\end{lemma}
\begin{proof}
    The procedure \compBasis{} calls \solveLPT{} on newly broken
    implicit sets of constraints $\CNR_1, \ldots, \CNR_m$, where
    $m\leq \sza$. By \lemref{same:cd}, the \LP-type problem defined by
    the ``meta'' constraints $\Impl(\CNR_1), \cdots, \Impl(\CNR_m)$ is
    an \LP-type problem of dimension $\lpdim$. We now interpret all
    the recursive calls to \compBasis of depth $i+1$, as being basis
    calculations for \solveLPT{}. Under this interpretation this is
    simply Seidel's algorithm. Now, \lemref{num:calls:new} readily
    implies both claims.
\end{proof}

\begin{lemma}%
    \lemlab{num:calls:new:2}%
    For an input of size $n$, the expected running time of \compBasis
    is $O\bigl( \DX{n} \bigr)$.
\end{lemma}
\begin{proof}
    Let $\TCBX{n}$ be the expected running time of \compBasis with an
    input of size $n$. Similarly, let $\TLPX{n}$ be the expected
    running time of \solveLPT{}, where each input set is of size at
    most $n$ (note that the number of input sets $m \leq \sza$).  We
    thus have that $\TCBX{ n} = O(1)$ for $n \leq O(1)$ and otherwise,
    we have that
    \begin{equation*}
        \TCBX{n} = O\bigl(  \DX{n}\bigr)  + \TLPX{n/\szb}.
    \end{equation*}
    Now, \lemref{levels} implies that
    \begin{align*}
      \TLPX{n}%
      &%
        =%
        O\pth{ \sza  \DX{n} + \TCBX{n}\log^\lpdim \sza \Bigr.}
        =%
        O\pth{ \sza  \DX{n} + \pth{\DX{n}  + \TLPX{n/\szb}\bigr.
        }\log^\lpdim \sza \Bigr.},      
    \end{align*}
    yielding the following recurrence, for some constant~$c$:
    \begin{equation*}
        \TLPX{n} \le%
        \pth{c\log^\lpdim\psi}\TLPX{n/\szb} + O\pth{ \sza  \DX{n}}.
    \end{equation*}
    Recall that Requirement \itemref{sub} (by \remref{replacement})
    implies that $\DX{n/\szb^i} = O(\DX{n}/\szb^{\eps i})$.  Expanding
    the recurrence gives
    \begin{align*}
      \TLPX{n}&%
                =%
                O\Bigl(
                \sum_{i=0}^\infty
                \pth{c \log^\lpdim \sza }^i      
                \sza  \DX{n/\szb^i}
                \Bigr)
      \\&%
      \leq%
      O\Bigl(
      \sza   \sum_{i=0}^\infty
      \Bigl[ \frac{c\log^\lpdim \sza}{\szb^\eps} \Bigr]^i      
      \DX{n}
      \Bigr)
      =%
      O\pth{\bigl.
      \sza
      \DX{n}
      },
    \end{align*}
    by \Eqref{weird}. To use the later, we might need to modify the
    given instance of implicit \LP as described in \lemref{weird}.
\end{proof}

We thus have proved our main theorem:

\begin{theorem}
    \thmlab{l:p:implicit}%
    Let $(\CN, \Objv)$ be an \LP{}-type problem of constant
    combinatorial dimension $\lpdim$. Let $\sza, \szb > 1$ be fixed
    constants. %
    For an input space $\Pi$, suppose that there is a function
    $\Impl : \Pi \to 2^\CN$ which maps inputs to
    constraints. Furthermore, assume that for any input $\PS \in \Pi$
    of size $n$, properties \itemref{base}--\itemref{sub} hold.
    Then a basis for $\Impl(\PS)$ can be computed in
    $O\bigl(\DX{n}\bigr)$ expected time.
\end{theorem}

\subsection{Some applications}
\seclab{tm:ex}

To illustrate the versatility of the new result, we briefly sketch a
few applications where some known results can be re-derived.

\subsubsection{Linear programming queries}

We first consider the problem of preprocessing a set $H$ of $n$
halfspaces in $\R^d$, so that we can quickly answer \emph{linear
   programming queries}, i.e., find a point in the intersection of $H$
maximizing any given linear function.  \Matousek~\cite{m-loq-93}
applied a multi-level parametric search to reduce the problem to
\emph{membership queries}: preprocess $H$ so that we can quickly
decide whether a query point lies in the intersection of $H$.  Our
technique easily gives a simpler randomized reduction:

\begin{corollary}
    If there is a data structure for membership queries with
    ${\cal P}(n)$ preprocessing time and $\DX{n}$ query time, then
    there is a data structure for linear programming queries with
    $O({\cal P}(n))$ preprocessing time and $O(\DX{n})$ query time,
    assuming that ${\cal P}(n)/n^{1+\eps}$ and $\DX{n}/n^\eps$ are
    monotonically increasing for some constant $\eps>0$.
\end{corollary}
\begin{proof}
    We build a data structure for linear programming queries for the
    given set $H$ as follows: arbitrarily divide $H$ into two subsets
    $H_1$ and $H_2$ of size $n/2$; store $H$ in the stated data
    structure for membership queries; recursively build a data
    structure for $H_1$ and for $H_2$.  The preprocessing time
    satisfies the recurrence
    ${\cal P}'(n) = 2{\cal P}'(n/2) + O({\cal P}(n))$, which yields
    ${\cal P}'(n)=O(\sum_i 2^i {\cal P}(n/2^i)) =O(\sum_i 2^i {\cal
       P}(n)/2^{(1+\eps)i})=O({\cal P}(n))$.

    To answer a linear programming query, we can immediately apply
    \thmref{l:p:implicit} with $\lpdim=d$.  
    Here, $\Pi$ consists of all sets $H$ that arise in the recursion,
    and we define $g(H)=H$.  When $H$ is divided into $H_1$ and $H_2$, we 
    have $g(H)=g(H_1)\cup g(H_2)$ trivially, and so we can set
    $\sza=\szb=2$.
\end{proof}

The above reduction is similar to an earlier reduction from ray
shooting queries to membership queries~\cite{c-garot-99}.  For linear
programming queries, similar results were obtained earlier by a
different randomized method by Chan~\cite{c-fdlpqme-96}.  The method
here uses randomization only in the query algorithm, not the
preprocessing, although the previous method can be derandomized more
effectively, as shown by Ramos~\cite{r-lpqr-00}.

\subsubsection{Minimum diameter of moving points}

As another example, consider the following problem: We are given a set
$P$ of $n$ linearly moving points $p_1, \ldots, p_n$ in $d$
dimensions, i.e., each $p_i$ is a function mapping a time value
$t\in\Re$ to a point $p_i(t)=a_i + b_i t$ for some $a_i,b_i\in\Re^d$.
We want to find a value~$t\in\Re$ that minimizes the diameter of the
point set at time $t$, i.e., that minimizes
$\max_{i,j}\|p_i(t)-p_j(t)\|$.

Gupta \etal\cite{gjs-facppimgo-96} applied parametric search to get an
$O(n\log^3n)$-time algorithm for the two-dimensional problem.
Clarkson~\cite{c-amdmpdcp} later described a randomized
$O(n\log n)$-time algorithm in dimension $d\le 3$, but this result
follows easily from \thmref{l:p:implicit}:

\begin{corollary}
    Given $n$ linearly moving points in two or three dimensions, we
    can find the time value that minimizes the diameter in
    $O(n\log n)$ expected time.
\end{corollary}
\begin{proof}
    For each pair of moving points $\{p_i,p_j\}$, define the following
    constraint in two variables $t$ and $y$:
    $\|p_i(t)-p_j(t)\|^2 \le y$.  Here, $y$ represents the square of
    the diameter.  Note that each such constraint forms a
    two-dimensional convex set.  Let $g(P)$ be the set of all
    $O(|P|^2)$ such constraints formed by all pairs in~$P$.  The
    problem is equivalent to minimizing $y$ over all $(t,y)\in\R^2$
    subject to the constraints in $g(P)$.  This is a convex program
    with combinatorial dimension $\lpdim=2$.

    Testing whether a given basis satisfies $g(P)$ amounts to testing
    whether a given $(t,y)$ satisfies $\|p_i(t)-p_j(t)\|^2 \le y$ for
    all $p_i,p_j\in P$, i.e., whether the diameter of the points
    $\{p_i(t): p_i\in P\}$ exceeds $\sqrt{y}$.  The diameter of a
    point set (at a fixed time) in dimension $d\le 3$ can be computed
    in $O(n\log n)$ time by known
    algorithms~\cite{cs-arscg-89,Ramos01}.  Thus, Property
    \itemref{verify} holds with $\DX{n}=O(n\log n)$.

    The divide property \itemref{partition} is easy to verify: As
    before, we can partition the point set $P$ into three subsets
    $P_1,P_2,P_3$ of equal size and express $g(P)$ as the union of
    $g(P_1\cup P_2)$, $g(P_2\cup P_3)$, and $g(P_1\cup P_3)$, with
    $\sza=3$ and $\szb=3/2$.
\end{proof}

\subsubsection{Inverse parametric minimum spanning trees}

Eppstein~\cite{e-spe-03} considered the following \emph{inverse
   parametric minimum spanning tree} problem: We are given a
connected, undirected, \emph{parametric} graph $G=(V,E)$ with $n$
vertices and $m$ edges, where the weight $w_e$ of each edge $e\in E$
is a linear function in $d$ variables (i.e., parameters).  We are also
given a spanning tree $T$.  The goal is to find values
$t_1, \ldots, t_d$ (if they exist) such that $T$ is the unique minimum
spanning tree (\MST) of $G$ when the $d$ variables are set to
$t_1,\ldots,t_d$.

\begin{corollary}
    The inverse parametric minimum spanning tree problem can be solved
    in $O(m)$ expected time for any constant~$d$.
\end{corollary}
\begin{proof}
    It is well known that in a (non-parametric) graph $G=(V,E)$, the
    tree $T$ is the unique \MST\ of $G$ if and only if for every
    non-tree edge $e=uv\in E-T$, every edge $e'$ on the path from $u$
    to $v$ in $T$ has smaller weight than $e$.

    For each pair of a non-tree edge $e=uv\in E-T$ and a tree edge
    $e'\in T$ such that $e'$ lies on the path from $u$ to $v$ in $T$,
    define a (linear) constraint
    $w_{e'}(t_1,\ldots,t_d)+z\le w_e(t_1,\ldots,t_d)$, where $z$ is an
    extra variable.  Let $g(G,T)$ be the set of all such $O(mn)$
    constraints.  The problem reduces to maximizing $z$ subject to the
    constraints in $g(G,T)$, and checking that the maximum is
    positive.  This is a linear program with combinatorial dimension
    $\lpdim=d+1$.

    Testing whether a given basis satisfies $g(G,T)$ amounts to
    testing whether $T$ is the unique minimum spanning tree of $G$
    after setting the $d$ variables to $t_1,\ldots,t_d$ and adding $z$
    to all tree edge weights.  This can be done in $O(m)$ expected
    time by Karger, Klein, and Tarjan's randomized \MST\
    algorithm~\cite{KargerKT95}, or more directly, by a known \MST\
    verification algorithm such as~\cite{King97}.  Thus, Property
    \itemref{verify} holds with $\DX{m}=O(m)$.

    To establish the divide property \itemref{partition}, we partition
    $E-T$ into two subsets $S_1$ and $S_2$ of equal size, and
    partition $T$ into two subsets $T_1$ and $T_2$ of equal size.  For
    each $i,j\in\{1,2\}$, define a graph $G_{ij}$ formed by keeping
    the edges from $S_i\cup T$ and then contracting all edges in
    $T-T_j$; similarly define the tree $T_j'$ formed by keeping the
    edges from $T$ and contracting all edges in $T-T_j$.  Then
    $g(G,T)$ is the union of $g(G_{ij},T_j')$ over all
    $i,j\in\{1,2\}$.  Thus, we have $\sza=4$ and $\eta=2$.
\end{proof}

In the journal version of his paper~\cite{e-spe-03}, Eppstein claimed
the above result by using the original optimization technique of
\cite{c-garot-99}, but with this less powerful technique, it is less
clear how to design an efficient decider.

\section{Tukey depth as an implicit \LP{}}
\seclab{tukey:depth}

\subsection{Finding a point of a given depth \TPDF{$k$}{k}}
\seclab{tm:k}

Let $\PS$ be a given non-degenerate set of $n$ points in $\Re^d$, and
let $k$ be a parameter.  Here, we consider the problem of finding a
point with Tukey depth at least~$k$, minimizing a linear function, if
such a point exists.

\subsubsection{Tukey depth, duality and levels}

Here, we provide some background on the Tukey depth (see
\defref{tukey:depth}).

\begin{remark}[Maximum Tukey depth of random points]
    \remlab{tueky:random}%
    Let $\PS$ be a set of $n$ random points sampled uniformly from the
    unit square. Setting $\eps = O(1/\sqrt{n})$, such a sample can be
    interpreted as an $\eps$-sample for the uniform measure of
    area. Specifically, it is known that a sample of size
    $O(1/\eps^2)$ is an $\eps$-sample, with some constant probability
    $\phi$ close to one \cite[Theorem 7.13]{h-gaa-11}, which readily
    implies that the point $(1/2,1/2)$ has Tukey depth
    $\geq n/2 - O(\sqrt{n})$ (with probability $\phi$).

    No point can have Tukey depth exceeding $n/2$, so this example is
    close to tight. By the centerpoint theorem, there is always a
    point in the plane of Tukey depth $n/3$, and in the worst case
    this is tight. As such, the maximum Tukey depth is always in the
    range $[n/3, n/2]$.
\end{remark}

The following characterization of all point of Tukey depth $k$ is well
known---we include a proof for the sake of completeness.

\begin{lemma}%
    \lemlab{depth:region}%
    Let $\HS$ be the set of all open halfspaces that contains strictly
    more than $n -k$ points of $\PS$ in them. Then
    $\TP_k = \bigcap_{\hplane^+ \in \HS} \hplane^+$ is the set of all
    points of Tukey depth $\geq k$.
\end{lemma}

\begin{proof}
    Consider a point $\pa$ of Tukey depth $k$, and assume, for the
    sake of contradiction, that $\pa \notin \TP_k$. But then, there
    exists an open halfspace $\hs^+ \in \HS$ that does not contain
    $\pa$. By construction $\hs^+$ contains $t > n-k$ points of
    $\PS$. Let $\hs$ denote the boundary hyperplane of $\hs^+$. Let
    $\hsa$ be the hyperplane passing through $\pa$ that is a
    translation of $\hs$. Clearly, the closed halfspace bounded by
    $\hsa$, that avoids $\hs^+$, contains $\pa$, but contains at most
    $n-t < k$ points of $\PS$. But this implies that the Tukey depth
    of $\pa$ is strictly smaller than $k$ (see \defref{tukey:depth}), a
    contradiction.

    As for the other direction, consider any point $\pa \in \TP_k$,
    and assume, for the sake of contradiction, that there is a
    halfspace $\hs^-$ that its boundary hyperplane passes through
    $\pa$, and it contains strictly less than $k$ points of $\PS$. By
    a small perturbation, one can assume that the boundary hyperplane
    $\hs$ does not contain any point of $\PS$. But then, the
    complement open halfspace $\hs^+$ contains strictly more than
    $n-k$ points of $\PS$, and it avoids $\pa$. As $\hs^+ \in \HS$, it
    follows that $\pa \notin \TP_k$. A contradiction.
\end{proof}

\begin{figure}
    \includegraphics[page=1,width=0.3\linewidth]%
    {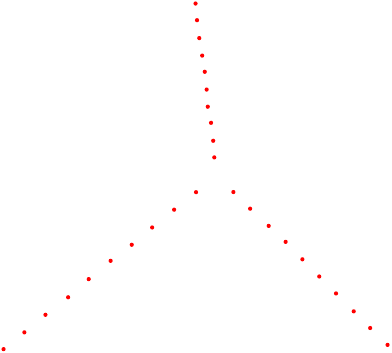}%
    \hfill%
    \includegraphics[page=2,width=0.3\linewidth]%
    {figs/30_points_no_sep_l_12_l_18}%
    \hfill%
    \includegraphics[page=3,width=0.3\linewidth]%
    {figs/30_points_no_sep_l_12_l_18}

    \caption{A set of thirty points, its dual, and the $12$ and $18$
       levels. These two levels cannot be separated by a line, as the
       maximum Tukey depth is $10$ (realized by the centerpoint).}
    \figlab{tukey:e:1}
\end{figure}

\begin{figure}
    \includegraphics[page=1,width=0.3\linewidth]%
    {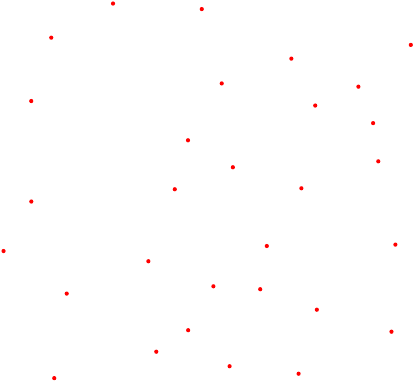}%
    \hfill%
    \includegraphics[page=2,width=0.3\linewidth]%
    {figs/30_r_points_sep_l_12_l_18}%
    \hfill%
    \includegraphics[page=3,width=0.3\linewidth]%
    {figs/30_r_points_sep_l_12_l_18}

    \caption{A set of $n=30$ random points, its dual, and the $12$ and
       $18$ levels. These two levels can be separated by a line, as
       the maximum Tukey depth is close to $n/2$ (see
       \remref{tueky:random}).}
    \figlab{tukey:e:2}
\end{figure}

Understanding this intersection polytope is somewhat easier in the
dual. See \defref{duality}. Let $\HSup$ (resp., $\HSdown$) be the set
of all hyperplanes that bounds from below (resp., above) a halfspace
which contains strictly more than $n - k$ points of $\PS$.  A point
$\pa \in \TP_k$ lies above (resp., below) all the hyperplanes of
$\HSup$ (resp., $\HSdown$).

The dual of $\HSup$ is the set of points
$\DHSup = \Set{ \Dual{\hp}}{\hp \in \HSup}$.  As duality is order
flipping (\factref{reverse}), it follows that all the points of
$\DHSup$ lie (vertically) below the hyperplane $\Dual{\pa}$. The set
$\DHSdown$ is defined in a similar fashion, and the points of
$\DHSdown$ lie above the hyperplane $\Dual{\pa}$.

Specifically, a point $\pa = (\pa_1, \ldots, \pa_d)$ induces in the
dual the hyperplane $\Dual{\pa}$ with equation
$x_d = -\pa_d + \sum_{i=1}^{d-1} x_i \pa_i$.  The condition that a
point $\pb \in \DHSup $ lies below the hyperplane $\Dual{\pa}$, then
reduces to the linear inequality (the variables being the coordinates
of $\pa$) that
\begin{equation*}
    \pb_d \leq -\pa_d + \sum_{i=1}^{d-1} \pb_i \pa_i,
\end{equation*}
As such, computing a point in $\TP_k$ is no more than solving a linear
program when each point of $\DHSup \cup \DHSdown$ induces a
constraint. This \LP computes a separating hyperplane between $\DHSup$
and $\DHSdown$.

In the dual, the set of points $\PS$ becomes a set of hyperplanes
$\Dual{\PS}$. A hyperplane $\hp \in \HSup$ in the dual is a point
$\Dual{\hp}$, such that there are at least $n-k$ hyperplanes of
$\Dual{\PS}$ strictly below it (in the $x_d$ direction). The number of
such hyperplanes is the \emph{level} of the point. The set $\HSup$ is
thus the set of all the points that have level strictly larger than
$n-k$. The boundary of the closure of $\DHSup$ (resp., $\DHSdown$) is
the \emphi{$(n-k)$-level} (resp., $k$-level).

In two dimensions, these levels are $k$-monotone polygonal curves. The
complexity of the $k$-level can be superlinear---a lower bound of
$n2^{\Omega(\sqrt{\log k})}$ is known \cite{t-psmks-01}, and currently
the best upper bound is $O(nk^{1/3})$~\cite{d-ibpksrp-98}. As such,
finding if there is a point of Tukey depth $k$ is equivalent to
deciding if there is a line that separates the $k$-level from the
$(n-k)$-level. In higher dimensions these levels are surfaces, and one
is looking for a hyperplane separating them.

In particular, this implies that to decide if there is a point of
Tukey depth $k$, one can solve the \LP defined in the dual to decide
if the $k$-level and the $(n-k)$-level are linearly separable. In two
dimensions, these two levels are polygonal curves, and it is enough to
decide if their vertices are linearly separable.  See
\figref{tukey:e:1} and \figref{tukey:e:2}.

Since the two levels can be computed in roughly $O(n^{4/3})$ time (and
this also bounds the number of their vertices), and linear programming
in two dimensions works in linear time, it follows that one can decide
whether there is a point of Tukey depth $k$ in roughly $O( n^{4/3})$
time. To get a faster algorithm, we deploy our implicit \LP
algorithm. Intuitively, the new algorithm computes an approximation of
the two levels, and thus computes a subset of the constraints of the
explicit LP\@. The algorithm refines these approximations in such a way
that the resulting rougher LP has a solution if and only if the
original LP has the same solution.

\subsubsection{The algorithm}

To deploy our framework for implicit \LP, we need a partition scheme
and a verifier, which we provide next.

\paragraph{Partitioning scheme}

A \emphi{bundle} $(\HS, \Simplex, \tau)$ of constraints in the
implicit \LP (of separating the two levels), is a simplex $\Simplex$,
the subset
\begin{equation*}
    \HS%
    =%
    \Set{ \hs \in \Dual{\PS}}{\hs \cap \Simplex \neq \emptyset},    
\end{equation*}
and a number $\tau$, which is the depth of a specific corner of
$\Simplex$ in the original arrangement $\Arr =
\ArrX{\Dual{\PS}}$. This bundle contains all the constraints that
define points of the $k$- and $(n-k)$-level of $\Arr$ that lie in
$\Simplex$. More broadly, given the bundle, one can compute the
arrangement $\Arr \cap \Simplex$. Thus, with the information provided,
one can compute the level of all the points inside $\Simplex$ in
$\Arr$ -- one can compute $\Arr \cap \Simplex$ and then propagate the
depth information from the specific corner of $\Simplex$.

A $(1/r)$-cutting restricted to the interior of $\Simplex$ can be
computed in $O( \cardin{\HS}r^{d-1})$ time, by \lemref{cuttings}. Each
cell in the cutting is a new bundle. Computing the depth of a point on
the boundary of each sub-simplex can be done easily in the same time
bound. Each bundle has $\cardin{\HS}/r$ hyperplanes in its conflict
list.

\paragraph{The verifier}

\begin{lemma}
    \lemlab{range:levels}%
    Let $\Dual{\PS}$ be a set of $n$ hyperplanes in $\Re^d$, and let
    $t^-, t^+$ be two numbers. Given a hyperplane $\hp$ one can decide
    in $O(n \log n + n^{d-1})$ time whether $\hp$ separates the
    $t^-$-level from the $t^+$-level of $\Arr = \ArrX{\Dual{\PS}}$.
\end{lemma}
\begin{proof}
    Consider any hyperplane of $\Dual{\PS}$ as bounding a halfspace
    that lies above it, and let $\HS$ be the resulting set of
    halfspaces. The \emphi{depth} of a point $\pa$ is the number of
    halfspaces of $\HS$ that contain $\pa$, and is the level of the
    point in the arrangement $\Arr$. Thus, it is enough to compute the
    range of depths realized by points lying on $\hp$. The
    intersection of each $d$-dimensional halfspace of $\HS$ with $\hp$
    is a $(d-1)$-dimensional halfspace of $\hp$. This range of depths
    can be computed by computing the arrangement of these $n$ induced
    halfspaces on the hyperplane $\hp$. As $\hp$ is $(d-1)$-dimensional,
    this requires $O(n \log n)$ time if $d=2$ (via essentially
    sorting), and $O(n^{d-1})$ in higher dimensions.

    Specifically, the algorithm computes for each face of the
    arrangement on $\hp$ its depth, which results in the range of
    depths of points on $\hp$. If the range lies outside $[t^-,t^+]$
    then $\hp$ is not the desired separating hyperplane.
\end{proof}

Consider a \emphi{bundle} $(\HS, \Simplex, \tau)$. Here, the algorithm
maintains the set $\HS \cap \Simplex$, and a corner $\pa$ of
$\Simplex$, such that its level is $\tau$. One can compute the range
of levels realized on the faces of $\Simplex$ using
\lemref{range:levels}. Since the level of a point is monotone
increasing along a vertical line, it follows that this provides the
range of levels realized also by the interior of $\Simplex$.

Given a hyperplane $\hp$, one needs to verify if it is feasible for
this bundle. To this end, one can first check if it intersects
$\Simplex$. If not, then the computed range of levels realized in
$\Simplex$ is sufficient to decide if it is feasible as far as the
levels inside $\Simplex$.  Otherwise, the algorithm computes the level
of some point $\pb \in \Simplex \cap \hp$, by computing how the level
changes as one moves from $\pa$ to $\pb$ (as we know the level of
$\pa$, and as $\pa \pb$ can intersect only hyperplanes in the set
$\HS \cap \Simplex$).  Now, using \lemref{range:levels} one can
compute the range of levels encountered on $\Simplex \cap \hp$, which
is sufficient to verify whether $\hp$ separates the desired levels
inside $\Simplex$. Note that if $m = \cardin{\HS \cap \Simplex}$, then
the running time of this algorithm is $O( m \log m + m^{d-1})$.

\paragraph{Putting everything together}

We next deploy the algorithm of \thmref{l:p:implicit}. In the dual, we
are solving an \LP computing a hyperplane separating the $k$ level
from the $(n-k)$ level. To this end, we pick an arbitrary linear
constraint on the \LP that we are trying to (say) minimize. We
sketched a verifier, and a partitioning schemes, that for $n$
hyperplanes, work in $\DX{n} = O(n \log n + n^{d-1})$ time. We thus
get the following.

\begin{theorem}
    \thmlab{tukey:geq:k}%
    Let $\PS$ be a set of $n$ points in $\Re^d$, and let $k$ be a
    parameter. One can compute a point in $\Re^d$ of Tukey depth at
    least $k$, if such a point exists, in $O(n \log n + n^{d-1})$
    expected time.
\end{theorem}

\subsubsection{Computing a point of maximum Tukey depth}
\seclab{tm:max}

Having solved the problem of deciding whether the maximum Tukey depth
is at least~$k$, we consider the problem of computing the maximum
Tukey depth.  Using binary search on top of \thmref{tukey:geq:k} one
can compute a point with maximum Tukey depth (paying an extra log
factor).  However, one can do better by considering the associated
ranking \LP problem (described in \secref{ranking:LP}).

For a vertex $\pa$ of $\ArrX{\Dual{\PS}}$, its associated rank is
$\rankX{\pa} = |n/2 - \levelX{\pa}|$. The \LP ranking problem, with
all the vertices of $\ArrX{\Dual{\PS}}$ as constraints, is exactly the
problem of finding a point of maximum Tukey depth. By
\lemref{ranking:lp} this is an \LP-type problem.  The idea is now to
solve this implicit \LP using the above machinery.

A basis now is a set of $d+1$ vertices of $\ArrX{\Dual{\PS}}$. It is
easy to modify the algorithm, for solving this implicit ranking \LP
problem, so that it explicitly computes the level of each
vertex/constraint being considered. As such, the rank of the basis is
known.  Now, such a basis defines a separating hyperplane and it is
supposed to separate specific levels as defined by its rank.  Thus,
the active constraints are induced by the points on these active
levels. Clearly, the verifier above works (unmodified) for this case,
as does the partition scheme. We thus get the following.

\begin{theorem}
    \thmlab{max:tukey}%
    Let $\PS$ be a set of $n$ points in $\Re^d$. One can compute a
    point in $\Re^d$ of maximum Tukey depth, in
    $O(n \log n + n^{d-1})$ expected time.
\end{theorem}

\subsubsection{Computing a depth region in two dimensions}
\seclab{depth:region}

In two dimensions, our approach can speed up an existing algorithm by
\Matousek~\cite{m-ccpps-90} for computing the entire region of depth
$\ge k$ (i.e., the convex polytope $\TP_k$ in
\lemref{depth:region}). The $O(n\log^4n)$ time bound is improved to
$O(n\log^2n)$.

\begin{theorem}
    Let $\PS$ be a set of $n$ points in $\Re^2$, and let $k$ be a
    parameter. One can compute the region $\TP_k$ of all points in
    $\Re^d$ of Tukey depth at least $k$ in $O(n \log^2 n)$ expected
    time.
\end{theorem}
\begin{proof}
    \Matousek~\cite{m-ccpps-90} noted that the problem reduces to
    computing the upper hull of the $k$-level in the dual plane.  He
    described a divide-and-conquer algorithm to compute this ``level
    hull'', using an oracle for the following subproblem: given a set
    of $n$ lines in the plane, a number $k$, and a vertical line
    $\ell$, find the tangent of the upper hull at $\ell$.  He applied
    a two-level parametric search to solve this subproblem in
    $O(n\log^3n)$ time.  By our approach, we can solve this subproblem
    in $O(n\log n)$ expected time: back in primal space, the
    subproblem is equivalent to finding a point in the region $\TP_k$
    that maximizes a given linear function.  This is an implicit \LP
    problem, and the method in this section yields an
    $O(n\log n)$-time randomized algorithm.

    The overall running time of \Matousek's divide-and-conquer
    algorithm is bounded by a logarithmic factor times the running
    time of the oracle, and is thus $O(n\log^2n)$.
\end{proof}

\section{The extremal yolk as an implicit \LP{}}
\seclab{extremal:yolk}

\subsection{Background}

\begin{definition}
    \deflab{extremal:yolk}%
    Let $\PS \subset \R^d$ be a set of $n$ points in general position.
    A median hyperplane is a hyperplane such that each of its two
    closed halfspaces contain at least $\ceil{n/2}$ points of $\PS$.
    A hyperplane is extremal if it passes through $d$ points of $\PS$.
    The \emphi{extremal yolk} is the ball of smallest radius
    interesting all extremal median hyperplanes of $\PS$.
\end{definition}

We give an $O(n^{d-1}\log n)$ expected time exact algorithm computing
the extremal yolk. To do so, we focus on the more general problem.

\begin{problem}
    \problab{k:sets}%
    Let $\EPlanesX{k}(\PS)$ be the collection of extremal hyperplanes
    which contain exactly $k$ points of $\PS$ on or above them.  Here,
    $k$ is not necessarily constant. The goal is to compute the
    smallest radius ball intersecting all hyperplanes of
    $\EPlanesX{k}(\PS)$.
\end{problem}

We observe that computing the extremal yolk can be reduced to the
above problem.

\begin{lemma}
    \lemlab{yolk:to:k:sets} The problem of computing the extremal yolk
    can be reduced to \probref{k:sets}.
\end{lemma}
\begin{proof}
    Suppose that $n$ is even, and define the set
    $\SE = \{n/2, n/2+1, \ldots, n/2+d\}$. A case analysis shows that
    any extremal median hyperplane $h$ must have exactly $m$ points of
    $\PS$ above or on $h$, where $m \in \SE$. Thus, computing the
    extremal yolk reduces to computing smallest radius ball
    intersecting all hyperplanes in the set
    $\bigcup_{m \in \SE} \EPlanesX{m}(\PS)$.

    \smallskip%
    When $n$ is odd, a similar case analysis shows that any extremal
    median hyperplane must have exactly $m$ points above or on it,
    where
    $m \in \SO = \{\ceil{n/2}, \ceil{n/2}+1, \ldots,
    \ceil{n/2}+d-1\}$.  Analogously, computing the extremal yolk with
    $n$ odd reduces to computing the smallest radius ball intersecting
    all hyperplanes in the set
    $\bigcup_{m \in \SO} \EPlanesX{m}(\PS)$.
\end{proof}

We use \thmref{l:p:implicit} to solve \probref{k:sets}. To this end,
we prove that \probref{k:sets} is an \LP-type problem when the
constraints are explicitly given (the following Lemma was also
observed by Bhattacharya \etal \cite{bjmr-oasirp-94}).

\begin{lemma}
    \lemlab{egg:lp:type}%
    \probref{k:sets} when the constraints (i.e., hyperplanes) are
    explicitly given is an \LP{}-type problem and has combinatorial
    dimension $\CD = d+1$.
\end{lemma}
\begin{proof}
    We prove something stronger, namely that the problem can be
    written as a linear program, implying it is an \LP{}-type
    problem. Let $\CN$ be the set of $n$ hyperplanes. For each
    hyperplane $\hplane \in \CN$, let
    $\DotProdY{a_\hplane}{x} + b_\hplane = 0$ be the equation
    describing $\hplane$, where $a_\hplane \in \R^d$,
    $\norm{a_\hplane} =1$, and $b_\hplane \in \R$.  Because of the
    requirement that $\norm{a_\hplane} = 1$, for a given point
    $\pa \in \R^d$, the distance from $\pa$ to a hyperplane $\hplane$
    is $\abs{\DotProdY{a_\hplane}{\pa} + b_\hplane}$.

    The linear program has $d+1$ variables and $2n$ constraints. The
    $d+1$ variables represent the center $\pa \in \R^d$ and radius
    $\rad \geq 0$ of the egg. The resulting \LP{} is
    \begin{align*}
      \text{min} \qquad%
      & \rad \\
      \text{subject to}\qquad
      & \rad \geq \DotProdY{a_\hplane}{\pa} +
        b_\hplane
      &\forall \hplane \in \CN \\
      & \rad \geq
        -\bigl(\DotProdY{a_\hplane}{\pa} +
        b_\hplane\bigr)
      & \forall \hplane \in \CN \\
      & \pa  \in \R^d.
    \end{align*}

    As for the combinatorial dimension, observe that any basic
    feasible solution for the above linear program will be tight for
    at most $d+1$ of the above $2n$ constraints.  Namely, these $d+1$
    hyperplanes are tangent to the optimal radius ball, and as such
    form a basis $\Basis \subseteq \CN$.
\end{proof}

To apply \thmref{l:p:implicit} we need to:
\begin{enumerate*}[label=(\roman*)]
    \item design an appropriate input space,
    \item develop a decider, and
    \item construct a constant number of subproblems which cover the
    constraint space.
\end{enumerate*}
As in \secref{tukey:depth}, the algorithm works in the dual space.
The following lemma shows that the dual of a ball $\Ball$ is the
closed region which lies between two branches of a hyperboloid. See
\figref{disk:dual}.

\begin{figure}
    \centering
    \includegraphics[scale=0.7]{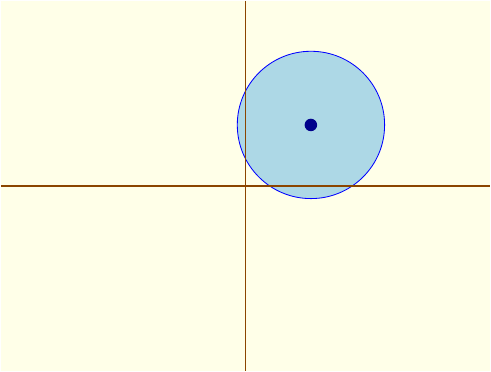}\hspace{20pt}%
    \includegraphics[page=2,scale=0.7]{figs/dual_of_a_disk}
    \caption{A disk and its dual.}
    \figlab{disk:dual}
\end{figure}

\begin{lemma}
    \lemlab{dual:ball} The dual of the set of points in a ball is the
    set of hyperplanes whose union forms the region enclosed between
    two branches of a hyperboloid.
\end{lemma}
\begin{proof}
    In $\R^d$ the hyperplane $\hplane$ defined by
    $x_d = \beta + \sum_{i=1}^{d-1} \alpha_i x_i$, or more compactly
    $\DotProdY{x}{(-\alpha,1)} = \beta$, intersects a disk $\Ball$
    centered at $\pa = (\pa_1, \ldots, \pa_d)$ with radius $r$ $\iff$
    the distance of $\hplane$ from $\pa$ is at most $r$. That is,
    $\hplane$ intersects $\Ball$ if
    \begin{align*}
      \frac{\cardin{\DotProdY{\pa}{(-\alpha,1)} -\beta }}
      {\norm{(-\alpha,1)}} \leq r
      &\iff
        \pth{\DotProdY{\pa}{(-\alpha,1)} -\beta }^2  \leq r^2
        \norm{(-\alpha,1)}^2 \\
      &\iff%
        \Bigl(\pa_d -\beta -\sum_{i=1}^{d-1} \alpha_i \pa_i\Bigr)^2  \leq r^2 (\norm{\alpha}^2 + 1).
      \\&%
      \iff
      \frac{\bigl(\pa_d -\beta -\sum_{i=1}^{d-1} \alpha_i \pa_i\bigr)^2}{r^2}  - \norm{\alpha}^2 \leq 1.
    \end{align*}
    The boundary of the above inequality is a hyperboloid in the
    variables $\pa_d - \beta - \sum_{i=1}^{d-1} \alpha_i \pa_i$ and
    $\alpha_1, \ldots, \alpha_{d-1}$. This corresponds to an affine
    image of a hyperboloid in the dual space $\alpha \times -\beta$.
\end{proof}

Throughout, let $\Dual{\Ball}$ denote the region between the two
branches of the hyperboloid dual to a ball $\Ball$.

\subsection{Solving the subproblem}
\seclab{subproblem}

We verify the requirements of \thmref{l:p:implicit} can be met.
First, we develop the algorithm for the violation test. As the
algorithm works in the dual, each subproblem consists of a simplex
$\Simp$, the dual set of hyperplanes $H = \Dual{\PS} \cap \Simp$
intersecting $\Simp$, and a parameter $\uDepth$ which is the number of
hyperplanes lying completely below $\Simp$. Additionally, the given
basis defines a ball $\Ball$, which in the dual is the region
$\Dual{\Ball}$. One can verify in the dual, that the violation test
must decide if there is a vertex of the $k$-level in the region
$\Re^d \setminus \Dual{\Ball}$.

\begin{lemma}
    \lemlab{e:y:d}%
    Given the input $(H, \Simp, \uDepth)$ and the region
    $\Dual{\Ball}$, checking whether there is a vertex of $\ArrX{H}$
    which has level $k$ and lies inside $\Re^d \setminus \Dual{\Ball}$
    can be done in $O(n^{d-1} \log n)$ time.
\end{lemma}
\begin{proof}
    Observe that $\Cell \cap (\R^d \setminus \Dual{\Ball})$ is the
    union of at most two convex regions. Indeed, the set
    $\R^d \setminus \Dual{\Ball}$ consists of two disjoint connected
    components, where each component is a convex body. Intersecting a
    simplex $\Cell$ with each component of
    $\R^d \setminus \Dual{\Ball}$ produces two (disjoint) convex
    bodies $\Cell'$ and $\Cell''$ (it is possible that $\Cell'$ or
    $\Cell''$ are empty). See \figref{simplex}.  Let $\Cell'$ be one
    of these two regions of interest. The algorithm will process
    $\Cell''$ in exactly the same way.

\begin{figure}[t]
    \centering%
    \includegraphics[page=1,scale=0.45]{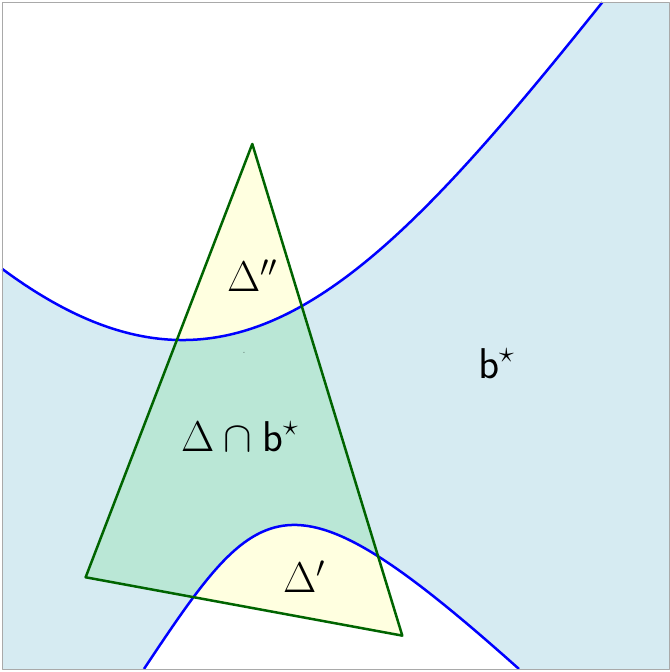}%
    \caption{The region $\Cell \cap (\R^d \setminus \Dual{\Ball})$
       consists of (at most) two disjoint convex regions, $\Delta'$
       and $\Delta''$.}
    \figlab{simplex}
\end{figure}

If $\Cell'$ is empty, then no constraints are violated. Otherwise, we
need to check for any violated constraints inside $\Cell'$. Let
$\BX{\Cell'}$ denote the boundary of $\Cell'$. Define $H' \subseteq H$
to be the subset of hyperplanes intersecting $\Cell'$. Observe that it
suffices to check if there is a vertex $v$ in the arrangement
$\ArrX{H'}$ such that:
\begin{enumerate*}[label=(\roman*)]
    \item $v$ has level $k$ in $\Dual{\PS}$,
    \item $v$ is a member of some cell in the zone
    $\ZoneY{\BX{\Cell'}}{H'}$, and
    \item $v$ is contained in $\Cell'$.
\end{enumerate*}

The algorithm computes $\ZoneY{\BX{\Cell'}}{H'}$. Next, it chooses a
vertex $v$ of the arrangement $\ArrX{H'}$ which lies inside $\Cell'$
and computes its level in $H'$ (adding $\uDepth$ to the count). The
algorithm then walks around the vertices of the zone \emph{inside}
$\Cell'$, computing the level of each vertex along the walk. Note that
the level between any two adjacent vertices in the arrangement differ
by at most a constant (depending on $d$). If at any point we find a
vertex of the desired level (such a vertex also lies inside $\Cell'$),
we report the corresponding median hyperplane which violates the given
ball $\Ball$. See \figref{walk} for an illustration.

\begin{figure}
    \centering%
    \includegraphics[page=3,scale=0.8]{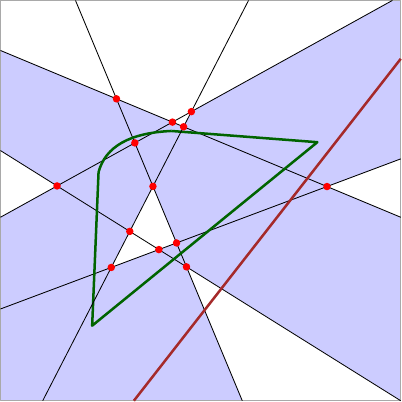}%
    \hspace{20pt}%
    \includegraphics[page=4,scale=0.8]{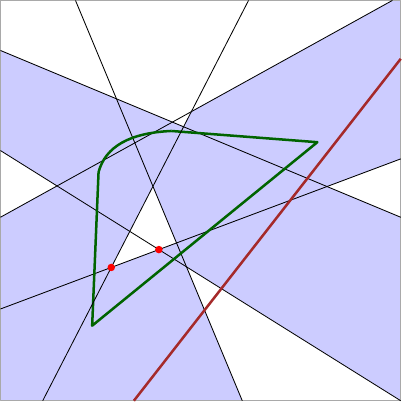}%
    \caption{Left: A convex region $\Cell'$, with one line lying
       completely below $\Cell'$ ($\uDepth = 1$). The shaded regions
       are the cells of $\ZoneY{\BX{\Cell'}}{H'}$. The vertices of the
       cells in the zone $\ZoneY{\BX{\Cell'}}{H'}$ are
       highlighted. Right: The vertices of $\ZoneY{\BX{\Cell'}}{H'}$
       which are part of the 3-level and contained inside $\Cell'$.}
    \figlab{walk}
\end{figure}

The running time of the algorithm is proportional to the complexity of
the zone $\ZoneY{\BX{\Cell'}}{H'}$. Because the boundary of $\Cell'$
is constructed from $d+1$ hyperplanes and the boundary of the
hyperboloid, \lemref{zone:complexity} implies that the zone complexity
is no more than $O(\cardin{H}^{d-1} \log \cardin{H})$.  As such, the
algorithm runs in time $\DX{n} = O(n^{d-1} \log n)$.
\end{proof}

\begin{lemma}
    \lemlab{k:sets} Let $\PS \subset \R^d$ be a set of $n$ points in
    general position.  For a given integer $k$, one can compute in
    $O(n^{d-1} \log n)$ expected time the smallest radius ball
    intersecting all of the hyperplanes of $\EPlanesX{k}(\PS)$
\end{lemma}
\begin{proof}
    Apply \thmref{l:p:implicit}. The violation test follows from
    \lemref{e:y:d}. Constructing the subproblems from a given input
    $(H, \Simp, \uDepth)$ is done in the same way as
    \thmref{tukey:geq:k}.  Compute a $(1/c)$-cutting (for constant $c$
    sufficiently large) of $H$ and clip the cutting inside
    $\Simp$. For each cell $\Simp_i$ in the cutting, compute
    $H_i = \Dual{\PS} \cap \Simp_i$ and the parameter $\uDepth_i$.
    This entire step can be performed in linear time.  Hence, the
    problem can be solved in $O(\DX{n})$ expected time, where
    $\DX{n} = O(n^{d-1} \log n)$.
\end{proof}

\begin{corollary}
    \corlab{range:k:sets}%
    Let $\PS \subset \R^d$ be a set of $n$ points in general position,
    and let $S \subseteq \IntRange{n}$. One can compute in
    $O(n^{d-1} \log n)$ expected time the smallest radius ball
    intersecting all of the hyperplanes of
    $\bigcup_{k \in S} \EPlanesX{k}(\PS)$.
\end{corollary}
\begin{proof}
    The algorithm is a slight modification of \lemref{k:sets}.  During
    the decision procedure, for each vertex in the zone, we check if
    it is a member of the $k$-level for some $k \in S$. If $S$ is of
    non-constant size, membership in $S$ can be checked in constant
    time using hashing.
\end{proof}

\subsection{Computing the extremal yolk and the egg}

\begin{theorem}
    \thmlab{discrete:yolk}%
    Let $\PS \subset \R^d$ be a set of $n$ points in general position.
    One can compute the extremal yolk of $\PS$ in $O(n^{d-1} \log n)$
    expected time.
\end{theorem}
\begin{proof}
    The result follows by applying \corref{range:k:sets} with the
    appropriate choice of $S$. When $n$ is even,
    \lemref{yolk:to:k:sets} tells us to choose
    $S = \{n/2, n/2+1, \ldots, n/2+d\}$. When $n$ is odd, we set
    $S = \{\ceil{n/2}, \ceil{n/2}+1, \ldots, \ceil{n/2}+d-1\}$.
\end{proof}

\begin{lemma}
    \lemlab{higher:dim:egg} Let $\PS \subset \R^d$ be a set of $n$
    points in general position.  One can compute the egg of $\PS$ in
    $O(n^{d-1} \log n)$ expected time.
\end{lemma}
\begin{proof}
    The proof follows by \corref{range:k:sets} with
    $S = \IntRange{n}$.  (Alternatively, by directly modifying the
    decision procedure to check if any vertex of the zone
    $\ZoneY{\Cell'}{H'}$ lies inside $\Cell'$.)
\end{proof}

\subsection{An algorithm sensitive to \TPDF{$k$}{k}}
\seclab{sensitive}

Here, we solve \probref{k:sets} for the case that $k \ll n$.  Recall
that to compute the extremal yolk, we reduced the problem to computing
the smallest ball intersecting all hyperplanes which contain a fixed
number of points of $\PS$ above or on them (see
\lemref{yolk:to:k:sets}). In particular, we developed an algorithm for
\probref{k:sets} and applied it when $k$ is proportional to $n$.

To develop an algorithm sensitive to $k$, we use the result of
\lemref{k:sets} as a black box and introduce the notion of shallow
cuttings.

\begin{definition}[Shallow cuttings]
    Let $H$ be a set of $n$ hyperplanes in $\R^d$. A
    \emphi{$k$-shallow cutting} is a collection of simplices such
    that:
    \begin{enumerate*}[label=(\roman*)]
        \item the union of the simplices covers the $(\leq k)$-level
        of $H$ (see \defref{levels}), and
        \item each simplex intersects at most $k$ hyperplanes of $H$.
    \end{enumerate*}
\end{definition}

\Matousek was the first to provide an algorithm for computing
$k$-shallow cuttings of size $O((n/k)^{\floor{d/2}})$ \cite{m-rph-92}.
When $d = 2, 3$, a $k$-shallow cutting of size $O(n/k)$ can be
constructed in $O(n\log n)$ time \cite{ct-odasc-16}.  For $d \geq 4$,
we sketch a randomized algorithm which computes a $k$-shallow cutting,
based on \Matousek's original proof of existence \cite{m-rph-92}.

\begin{lemma}[Proof sketch in \apndref{sc:alg}]
    \lemlab{sc:alg} Let $H$ be a set of $n$ hyperplanes in $\R^d$. A
    $k$-shallow cutting of size $O((n/k)^{\floor{d/2}})$ can be
    constructed in $O(k(n/k)^{\floor{d/2}} + n\log n)$ expected time.
    For each simplex $\Cell$ in the cutting, the algorithm returns the
    set of hyperplanes intersecting $\Cell$ and the number of
    hyperplanes lying below $\Cell$.
\end{lemma}

Let $\PS \subset \R^d$ be a set of $n$ points and let $H = \Dual{\PS}$
be the set of dual hyperplanes.  The algorithm
first computes a $k$-shallow cutting for the top and bottom $(\le k)$-levels
for the given set of hyperplanes $H$ using \lemref{sc:alg}.  Let
$\Cell_1, \ldots, \Cell_\ell$, where
\begin{equation*}
    \ell = O((n/k)^{\floor{d/2}}),
\end{equation*}
be the collection of simplices in the cutting. For each simplex
$\Cell_i$, we have the subset $H \cap \Cell_i$ and the number of
hyperplanes lying completely below $H$ (which is at most $k$). For
each cell $\Cell_i$, let $\Impl(\Cell_i)$ be the set of vertices of
$\ArrX{H}$ which have level $k$ or $n-k$ and are contained in
$\Cell_i$.

\subsubsection{The algorithm} %
The algorithm computes the above shallow cutting, and treats each
simplex as a bundle.  We now apply the algorithm of
\secref{lp:bundles} to solve the batched \LP problem defined by these
bundles, except that we delegate each basis calculation/verification
to a call to the algorithm of \lemref{k:sets}, which involves a single
bundle (i.e., $k$ hyperplanes).

\begin{lemma}
    \lemlab{yolk:small:k}%
    Let $\PS \subset \R^d$ be a set of $n$ points in general position.
    For a given integer $k$, one can compute in
    $O\pth{k^{d-1}\big(1 + (n/k)^{\floor{d/2}}\big) \log k + n \log n
    }$ expected time the smallest radius ball intersecting all of the
    hyperplanes of $\EPlanesX{k}(\PS)$.
\end{lemma}
\begin{proof}
    The algorithm is described above. The correctness is immediate
    from \lemref{num:calls:new}.  As for the running time, the
    algorithm performs $O( \ell)$ basis calculations and violation
    tests, and each one takes $O(k^{d-1} \log k)$ time, as this is the
    size of the conflict list of each bundle. This running time
    dominates the time to compute the cutting, except for the
    $O(n \log n)$ additive term.
\end{proof}

\section{The (continuous) yolk as an implicit \LP{}}
\seclab{continuous:yolk}

\subsection{Background}
\begin{definition}
    Let $\PS \subset \R^d$ be a set of $n$ points in general position.
    The \emphi{continuous yolk} of $\PS$ is the ball of smallest
    radius intersecting all median hyperplanes of $\PS$.
\end{definition}

In contrast to \defref{extremal:yolk}, we emphasize that the
(continuous) yolk must intersect all median hyperplanes defined by
$\PS$ (not just extremal median hyperplanes).

As before, the algorithm works in the dual space. For an integer $k$,
let $\PlanesX{k}(\PS)$ be the collection of halfspaces containing
exactly $k$ points of $\PS$ on or above it. Equivalently, $\Dual{\PS}$
is the collection of hyperplanes defined by $\PS$ in the dual space,
and $\Dual{\bigl(\PlanesX{k}(\PS)\bigr)}$ is the $k$-level of
$\Dual{\PS}$. Our problem can be restated in the dual space as
follows.

\begin{problem}
    \problab{cont:yolk:dual} Let $\PS$ be a set of points in $\R^d$ in
    general position and let $k$ be a given integer. Compute the ball
    $\Ball$ of smallest radius so that all points in the $k$-level of
    $\Dual{\PS}$ are contained inside the region $\Dual{\Ball}$.
\end{problem}

Let $\lvY{k}{\PS} = \Dual{\bigl(\PlanesX{k}(\PS)\bigr)}$ denote the
set of all points in the $k$-level of $\Dual{\PS}$.  Note that
$\lvY{k}{\PS}$ consists of points which are either contained in the
interior of some $\ell$-dimensional flat, where
$0 \leq \ell \leq d-1$, or in the interior of some $d$-dimensional
cell of $\ArrX{\Dual{\PS}}$.

We take the same approach as the algorithm of
\thmref{discrete:yolk}---building a decider subroutine, and showing
that the input space can be decomposed into subproblems efficiently.
However the problem is more subtle, as the collection of constraints
(i.e., median hyperplanes) is no longer a finite set.

\subsubsection{The input space} %
The input consists of a simplex $\Cell$. The algorithm, in addition to
$\Cell$, maintains the set of hyperplanes
\begin{equation*}
    H = \Dual{\PS} \cap \Cell = \Set{\hplane \in
       \Dual{\PS}}{\hplane \cap \Cell \neq \emptyset},
\end{equation*}
and a parameter $\uDepth$ which is equal to the number of hyperplanes
of $\Dual{\PS}$ lying completely below $\Cell$.

\subsubsection{The implicit constraint space} %
Each input $\Cell$ maps to a region $R$ which is the portion of the
$k$-level $\lvY{k}{\PS}$ contained inside $\Cell$.  For each
$d$-dimensional cell in $R$, we compute its bottom-vertex
triangulation (see, e.g., \cite[Section 6.5]{m-ldg-02}), and collect
all of these simplices, and all lower-dimensional faces of $R$, into a
set $\Impl(\Cell)$. See \figref{bvt}.

\begin{figure}[t]
    \phantom{}
    \hfill
    \includegraphics[width=0.28\linewidth,page=1]{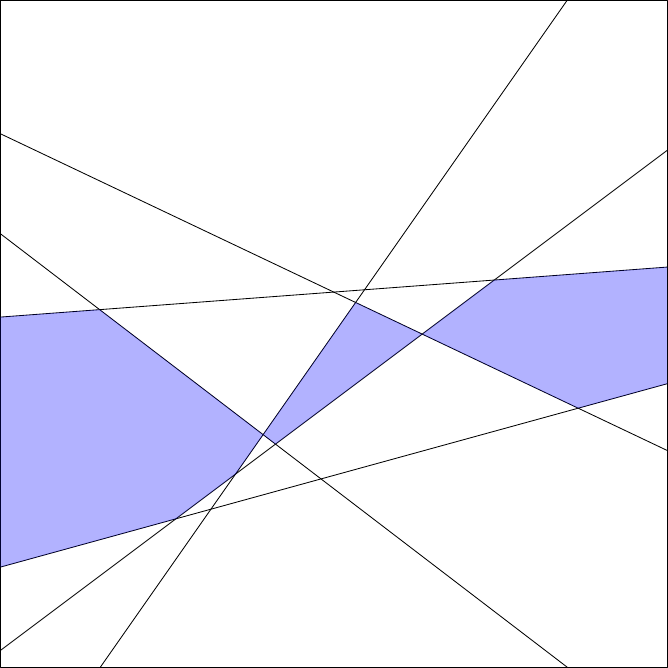}%
    \hfill
    \includegraphics[width=0.28\linewidth,page=2]{figs/k_level}%
    \hfill
    \includegraphics[width=0.28\linewidth,page=3]{figs/k_level}%
    \hfill
    \phantom{}

    \caption{Left: A set of lines and the cells of the $3$-level.
       Middle: A simplex $\Cell$, with the portion of the $3$-level
       inside $\Cell$. Right: Triangulating the portion of the
       $3$-level contained inside $\Cell$. All red triangles together
       with the lower-dimensional faces of the $3$-level form the set
       of constraints $\Impl(\Cell)$.}
    \figlab{bvt}
\end{figure}

Let $\Xi$ be the collection of all simplices formed from $d+1$
vertices of the arrangement $\ArrX{\Dual{\PS}}$.  We let $\CN$ be the
union of the sets $\Impl(\Cell)$ over all simplices $\Cell \in \Xi$.
To see why this suffices, each simplex in the input space is a simplex
generated by a cutting algorithm.  One property of cutting algorithms
\cite{c-chdc-93} is that the simplices returned are induced by
hyperplanes of $\Dual{\PS}$. Indeed, each simplex has (at most) $d+1$
vertices, and upon inspection of the cutting algorithm, each vertex is
defined by $d$ hyperplanes of $\Dual{\PS}$. There are a finite number
of simplices $\Cell$ to consider, and each $\Cell$ induces a fixed
subset of constraints $\Impl(\Cell) \subseteq \CN$.

As such, $\CN$ forms our constraint set, where each constraint is of
constant size (depending on $d$).  Clearly, a solution satisfies all
constraints of $\CN$ if and only if the solution intersects all
hyperplanes in the set $\PlanesX{k}(\PS)$. For a given subset
$\CNB \subseteq \CN$, the objective function is the minimum radius
ball $\Ball$ such that all regions of $\CNB$ are contained inside the
region $\Dual{\Ball}$. In particular, the problem of computing the
minimum radius ball $\Ball$ such that $\Dual{\Ball}$ contains all
points of $\lvY{k}{\PS}$ in its interior is an \LP{}-type problem of
constant combinatorial dimension.

\subsection{The algorithm}
\subsubsection{Constructing subproblems} %
For a given input simplex $\Cell$ (along with the set
$H = \Dual{\PS} \cap \Cell$ and the number $\uDepth$) a collection of
subproblems $\Cell_1, \ldots, \Cell_\sza$ (with the corresponding sets
$H_i$ and numbers $\uDepth_i$ for $i = 1, \ldots, \sza$) can be
constructed as described in \lemref{k:sets}, by computing a cutting of
the planes $H$ and clipping this cutting inside $\Cell$.  In
particular, we have that $\bigcup_i \Impl(\Cell_i) =
\Impl(\Cell)$. Strictly speaking, we have not decomposed the
constraints of $\Impl(\Cell)$ (as required by \thmref{l:p:implicit}),
but rather have decomposed the region which is the union of the
constraints of $\Impl(\Cell)$.  This step is valid, as a solution
satisfies the constraints of $\bigcup_i \Impl(\Cell_i)$ if and only if
it satisfies the constraints of $\Impl(\Cell)$.

\subsubsection{The decision procedure} %
Given a candidate solution $\Dual{\Ball}$, the problem is to decide if
$\Dual{\Ball}$ contains $\Impl(\Cell)$ in its interior.  The decision
algorithm itself is similar as in the proof of \thmref{discrete:yolk}.
Consider the set $\Cell \cap (\R^d \setminus \Dual{\Ball})$, where
$\Cell$ is a simplex, and observe that it is the union of at most two
convex regions. Let $\Cell'$ be one of these two regions of
interest. Observe that it suffices to check if there is a point on the
boundary of $\Cell'$ which is part of the $k$-level.  Let
$H' \subseteq H$ be the subset of hyperplanes intersecting $\Cell'$.

To this end, compute $\ZoneY{\BX{\Cell'}}{H'}$. For each
$(d-1)$-dimensional face $f$ of $\Cell'$, the collection of regions
$\Xi = \Set{f \cap s}{s \in \ZoneY{\BX{\Cell'}}{H'}}$ forms a
$(d-1)$-dimensional arrangement restricted to $f$.  Furthermore, the
complexity of this arrangement lying on $f$ is at most
$O(n^{d-1} \log n)$.  Notice that the level of all points in the
interior of a face of $\Xi$ is constant, and two adjacent faces
(sharing a boundary) have their level differ by at most a
constant. The algorithm picks a face in $\Xi$, computes the level of
an arbitrary point inside it (adding $\uDepth$ to the count). Then,
the algorithm walks around the arrangement, exploring all faces, using
the level of neighboring faces to compute the level of the current
face. If at any step a face has level $k$, we report that the input
$(\Cell, H, \uDepth)$ violates the candidate solution $\Dual{\Ball}$.

\paragraph{Analysis of the decision procedure} %
We claim the running time of the algorithm is proportional to the
complexity of the zone $\ZoneY{\BX{\Cell'}}{H'}$. Indeed, for each
$(d-1)$-dimensional face $f$ of $\Cell'$ (where $f$ may either be part
of a hyperplane or part of the boundary of $\Dual{\Ball}$), we can
compute the set $\Set{f \cap s}{s \in \ZoneY{\BX{\Cell'}}{H'}}$ in
time proportional to the total complexity of $\ZoneY{\BX{\Cell'}}{H'}$
(assuming we can intersect a hyperplane with a portion of a
constant-degree surface efficiently). The algorithm then computes the
level of an initial face naively in $O(\cardin{H'})$ time, and
computing the level of all other faces can be done in
$O(\cardin{\ZoneY{\BX{\Cell'}}{H'}})$ time by performing a graph
search on the arrangement.

Because the boundary of $\Cell'$ is constructed from $d+1$ hyperplanes
and the boundary of the hyperboloid, \lemref{zone:complexity} implies
that the zone complexity is bounded by
$O(\cardin{H}^{d-1} \log \cardin{H})$.  As such, our decision
procedure runs in time $\DX{n} = O(n^{d-1} \log n)$.

\paragraph{A slightly improved decision procedure} %
In $\R^3$, one can shave the $O(\log n)$ factor to obtain an $O(n^2)$
expected time algorithm. We modify the decision procedure as follows,
which avoids computing the zone $\ZoneY{\BX{\Cell'}}{H'}$. For each 2D
face $f$ of $\Cell'$, simply compute the arrangement of the set of
curves $\Set{f \cap h}{h \in H}$ on $f$ in $O(n^2)$ time. As before,
we perform a graph search on this arrangement, computing the level of
each face. If at any time we discover a point on the boundary of
$\Cell'$, of the desired level, we report that the given input
violates the given candidate solution.

For higher dimensions $d\ge 4$, we can similarly avoid computing the
zone.  Recall that the goal is to find a point $p$ that lies in the
intersection of the $k$-level with a $(d-1)$-dimensional face $f$ of
$\Cell'$.  Consider the (unknown) cell $\gamma$ containing $p$ in the
arrangement of $\Set{f \cap h}{h \in H}$ on $f$.  Imagine moving $p$
to lie in an arbitrarily small neighborhood of the minimum point $p'$
in $\gamma$ with respect to the $x_d$ coordinate.  The level of $p$
remains unchanged by the move (and differs from the level of $p'$ by
at most $d$).
	
\begin{itemize}
    \item \emph{Case 1}: $p'$ is incident to at most $d-2$ hyperplanes
    in $H$.  We can search for such a $p'$, by trying all tuples of at
    most $d-2$ hyperplanes.  For each such tuple $(h_1,\ldots,h_\ell)$
    with $\ell\le d-2$, we compute all $O(1)$ local $x_d$-minima $\pq$
    of $f\cap h_1\cap\ldots \cap h_\ell$. Next, we compute the level
    of $q$ naively in $O(n)$ time. Finally, we examine the
    neighborhood of $q$.  This takes $O(n^{d-2}\cdot n)=O(n^{d-1})$
    time.
	
    \item \emph{Case 2}: $p'$ is incident to $d-1$ hyperplanes in $H$.
    We try all tuples of $d-3$ hyperplanes.  For each such tuple
    $(h_1,\ldots, h_{d-3})$, we compute the two-dimensional
    arrangement of the set of curves
    \begin{equation*}
        \Set{f\cap h_1\cap\ldots\cap h_{d-3}\cap h}{h\in H} 
    \end{equation*}
    on the two-dimensional surface $h_1\cap\ldots\cap h_{d-3}\cap f$
    in $O(n^2)$ time. As above, we perform a graph search on this
    arrangement, computing the level of each cell and each vertex in
    the arrangement, and examine the neighborhood of each vertex. The
    total time is $O(n^{d-3}\cdot n^2)=O(n^{d-1})$.
\end{itemize}
   
Thus, our improved decision procedure runs in time
$\DX{n} = O(n^{d-1})$ for any $d\ge 3$.

\begin{lemma}
    \probref{cont:yolk:dual} can be solved in $O(n\log n +
    n^{d-1})$ %
    expected time, where $n = \cardin{\PS}$.
\end{lemma}
\begin{proof}
    Follows by plugging the above discussion into
    \thmref{l:p:implicit}.
\end{proof}

By modifying the decision procedure appropriately, we also obtain a
similar result to \corref{range:k:sets}.

\begin{corollary}
    \corlab{range:cont} Let $\PS \subset \R^d$ be a set of $n$ points
    in general position, and let $S \subset \IntRange{n}$. The
    smallest ball intersecting all hyperplanes in
    $\bigcup_{k \in S} \PlanesX{k}(\PS)$ can be computed in
    $O(n\log n + n^{d-1})$ %
    expected time.
\end{corollary}

\begin{theorem}
    \thmlab{cont:yolk} Let $\PS \subset \R^d$ be a set of $n$ points
    in general position.  One can compute the yolk of $\PS$ in
    $O(n\log n + n^{d-1})$ %
    expected time.
\end{theorem}
\begin{proof}
    The result follows by applying \corref{range:cont} with the
    appropriate choice of $S$. When $n$ is even,
    \lemref{yolk:to:k:sets} tells us to choose
    $S = \{n/2, n/2+1, \ldots, n/2+d\}$. When $n$ is odd, we set
    $S = \{\ceil{n/2}, \ceil{n/2}+1, \ldots, \ceil{n/2}+d-1\}$.
\end{proof}

\section{The Tukey ball and center ball as implicit \LP{}s}
\seclab{tukey:center}

Here, we are dealing with an extension of Tukey depth
(\defref{tukey:depth}). The set of all points with Tukey depth
$\geq k$ is the polytope $\TP_k$ (see \lemref{depth:region}).  Recall
that by the centerpoint theorem $\TP_{n/(d+)}$ is not empty.

\begin{definition}
    \deflab{tukey:ball} Let $\PS \subset \R^d$ be a set of $n$ points
    in general position.  For a parameter $k \leq n$, the \emphi{Tukey
       ball} of $\PS$ is the smallest radius ball intersecting
    halfspaces in the set $\PlanesX{k}(\PS)$.
\end{definition}

The Tukey median is a point in $\R^d$ with maximum Tukey depth.  If
the Tukey median of $\PS$ has Tukey depth $k(\PS)$, then for
$k > k(\PS)$ the set $\TP_k$ is empty---the Tukey ball has non-zero
radius. When $k \leq k(\PS)$, $\TP_k$ is non-empty, implying that the
Tukey ball has radius zero.

\begin{definition}
    \deflab{center:ball} Let $\PS \subset \R^d$ be a set of $n$ points
    in general position.  For a parameter $k \leq k(\PS)$, the
    \emphi{center ball} of $\PS$ is the ball of largest radius
    contained in the region $\TP_k$.
\end{definition}

Recently, Oh and Ahn \cite{oa-ccr-19} develop an $O(n^2\log^4 n)$ time
algorithm for computing the polytope $\TP_k$ in $\R^3$.  In contrast,
the center ball is the largest ball contained inside $\TP_k$, and we
show it can be computed in expected time $O(n^2 \log n)$.

\subsection{The Tukey ball in the dual}

For a set of $n$ points $\PS$ in general position, it suffices to
restrict our attention to hyperplanes which contain $d$ points of
$\PS$, and one of the open halfspaces contains more than $n-k$ points
of $\PS$. In the dual, each point $\pa \in \PS$ is mapped to a
hyperplane $\Dual{\pa}$ (see \defref{duality}). A hyperplane $h$
passing through $d$ points of $\PS$ maps to a point $\Dual{h}$ which
is a vertex in the arrangement $\ArrX{\Dual{\PS}}$.

Recall that by \lemref{dual:ball}, a ball $\Ball$ in the primal maps
to the region enclosed by two branches of a hyperboloid.  Formally,
the region $\Dual{\Ball}$ is the collection of points
$(x_1, \ldots, x_d) \in \R^d$ satisfying has the equation
$(x_d/\alpha_d)^2 - \sum_{i=1}^{d-1} (x_i/\alpha_i)^2 \leq 1$, where
$\alpha_1, \ldots, \alpha_d \in \R$ define the hyperboloid, and are
determined by $\Ball$. We say that a point $(x_1, \ldots, x_d)$ lies
above the top branch of $\Dual{\Ball}$ if the inequality
\begin{equation*}
    x_d \geq \alpha_d \sqrt{1 + \sum_{i=1}^{d-1} (x_i/\alpha_i)^2}    
\end{equation*}
holds. A point lying below the bottom branch of $\Dual{\Ball}$ is
defined analogously.

Let $h$ be a hyperplane. Suppose the open halfspace $h^-$ below $h$
contains $k$ points of $\PS$. In the dual, a vertical ray $\ray_h$
shooting upwards from the point $\Dual{h}$ intersects $k$ hyperplanes
of $\Dual{\PS}$. When a hyperplane $h$ intersects $\Ball$ in its
interior, then $\Ball \cap h^- \neq \varnothing$ and
$\Ball \not\subseteq h^-$.  In the dual, $\Dual{\Ball}$ contains the
point $\Dual{h}$, and the upward ray $\ray_h$ intersects the boundary
of $\Dual{\Ball}$ once.  Alternatively, if $\Ball \subseteq h^-$, then
in the dual the upward ray $\ray_h$ intersecting the boundary of
$\Dual{\Ball}$ twice (once each at the top and bottom branch). As
such, if $h^-$ is an open halfspace containing $k$ points of $\PS$
below it and does not intersect $\Ball$, then the upward ray $\ray_h$
does not intersect the boundary of $\Dual{\Ball}$. Hence, $\ray_h$
must lie entirely above the top branch of $\Dual{\Ball}$.  See
\figref{tukey:ball}.

\begin{figure}
    \centering \includegraphics[scale=0.8, page=1]{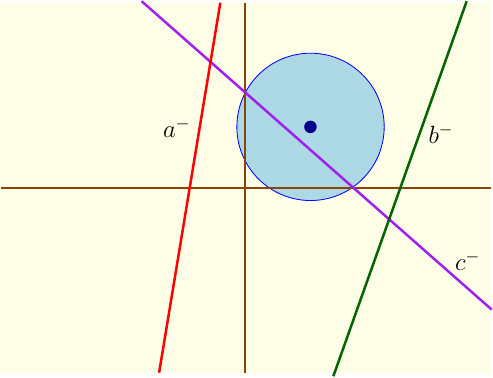}
    \hspace{20pt}%
    \includegraphics[scale=0.8, page=2]{figs/tukey_ball}
    \caption{A ball and three lines. Each line induces a halfspace
       which lies below the line. In the dual, this corresponds to
       three vertically upward rays.}
    \figlab{tukey:ball}
\end{figure}

Summarizing the above discussion, the problem of computing the Tukey
ball is equivalent to the following.

\begin{problem}
    \problab{dual:tukey}%
    Let $\PS \subset \R^d$ be a set of $n$ points in general
    position. The goal is to compute the ball $\Ball$, of smallest
    radius, such that (recalling \defref{top:bottom:lvls}):
    \begin{compactenumI}
        \item each point of the top $k$-level $\TY{k}{\Dual{\PS}}$,
        the vertical upward ray intersects $\Dual{\Ball}$, and
        \item each point of the bottom $k$-level $\BY{k}{\Dual{\PS}}$,
        the vertical downward ray intersects $\Dual{\Ball}$.
    \end{compactenumI}
\end{problem}

\begin{lemma}
    \lemlab{tukey:ball} Let $\PS \subset \R^d$ be a set of $n$ points
    in general position and $k \leq n$ a parameter. The Tukey ball can
    be computed in $O(n^{d-1} \log n)$ expected time.
\end{lemma}
\begin{proof}
    The proof uses \thmref{l:p:implicit} to solve the dual problem
    (this problem is \LP{}-type with constant combinatorial dimension,
    where the constant depends on $d$). The input consists of a
    simplex $\Cell$, the set of hyperplanes
    $H = \Dual{\PS} \cap \Cell$ intersecting $\Cell$, and the number
    of hyperplanes lying above and below $\Cell$. A given input can be
    decomposed using cuttings, as in the algorithms for
    \thmref{discrete:yolk}, and \thmref{cont:yolk}.
  
    We sketch the decision procedure. Given a candidate ball $\Ball$,
    we want to decide if $\Dual{\Ball}$ violates any constraints
    induced by $H$. Equivalently, $\Dual{\Ball}$ is an invalid
    solution if either condition holds:
    \begin{enumerate*}[label=(\roman*)]
        \item there is a element of $\TY{k}{\Dual{\PS}}$ which is
        above the top branch of $\Dual{\Ball}$, or
        \item there is a element of $\BY{k}{\Dual{\PS}}$ which is
        below the bottom branch of $\Dual{\Ball}$.
    \end{enumerate*}
    As such, a straightforward modification of the decision procedure
    described in \lemref{e:y:d} yields a decider in
    $O(\cardin{H}^{d-1} \log \cardin{H})$ expected time.
\end{proof}

\subsubsection{Improved algorithm}

\begin{lemma}
    \lemlab{tukey:ball:k} Let $\PS \subset \R^d$ be a set of $n$
    points in general position and $k \leq n$ a parameter. The Tukey
    ball can be computed in
    $\tldO\pth{k^{d-1}\big(1 + (n/k)^{\floor{d/2}}\big)}$ expected
    time.
\end{lemma}
\begin{proof}
    The algorithm is the same as described in \lemref{yolk:small:k}
    with a small change: compute a shallow cutting for the top
    $(\leq k)$-level and bottom $(\leq k)$-level of $\Dual{\PS}$. Now
    run the randomized incremental algorithm of \lemref{yolk:small:k}
    on these collection of simplices with \lemref{tukey:ball} as a
    black box to solve the subproblems of smaller size.
\end{proof}

\subsection{The center ball in the dual}

For a parameter $k$, recall that our goal is to compute the largest
ball which lies inside all open halfspaces containing more than $n-k$
points of $\PS$. From the discussion above, in the dual this
corresponds to the following problem.

\begin{problem}
    \problab{center:ball} Let $\PS \subset \R^d$ be a set of $n$
    points in general position. The goal is to compute the ball
    $\Ball$ of largest radius such that:
    \begin{compactenumI}
        \item each point of the top $k$-level $\TY{k}{\Dual{\PS}}$
        lies below the bottom branch of $\Dual{\Ball}$, and
        \item each point of the bottom $k$-level $\BY{k}{\Dual{\PS}}$
        lies above the top branch of $\Dual{\Ball}$.
    \end{compactenumI}
\end{problem}

\begin{lemma}
    \lemlab{center:ball} Let $\PS \subset \R^d$ be a set of $n$ points
    in general position and $k \leq n$ a parameter. The center ball
    can be computed in $O(n^{d-1} \log n)$ expected time.
\end{lemma}
\begin{proof}
    As usual, we use \thmref{l:p:implicit} to solve the dual problem
    (this problem is \LP{}-type with constant combinatorial dimension,
    where the constant depends on $d$). The input consists of a
    simplex $\Cell$, the set of hyperplanes
    $H = \Dual{\PS} \cap \Cell$ intersecting $\Cell$, and the number
    of hyperplanes lying above and below $\Cell$.  A given input can
    be decomposed using cuttings, as used in previous algorithms.

    We sketch the decision procedure. We are also given a candidate
    ball $\Ball$.  The algorithm computes the zone
    $\ZoneY{\BX{\Cell}}{H}$ and computes the level of each vertex of
    $\ZoneY{\BX{\Cell}}{H}$ inside $\Cell$ (taking into account the
    number of hyperplanes above and below $\Cell$).  If we find a
    vertex of either the top or bottom $k$-level which also lies
    inside $\Dual{\Ball}$, we report the violated constraint.
    Otherwise, if we find a vertex of the top $k$-level lying above
    the top branch of $\Dual{\Ball}$ or a vertex of the bottom
    $k$-level lying below the bottom branch of $\Dual{\Ball}$, then
    the solution $\Ball$ is also deemed infeasible.  This decision
    procedure can be implemented in
    $O(\cardin{H}^{d-1} \log \cardin{H})$ expected time.
\end{proof}

\subsubsection{Improved algorithm}

\begin{lemma}
    \lemlab{center:ball:k} Let $\PS \subset \R^d$ be a set of $n$
    points in general position and $k \leq n$ a parameter. The center
    ball can be computed in
    $\tldO\pth{k^{d-1}\big(1 + (n/k)^{\floor{d/2}}\big)}$ expected
    time.
\end{lemma}
\begin{proof}
    The same argument for \lemref{tukey:ball:k} applies here, using
    \lemref{center:ball} as a black box to solve the subproblems
    generated by the $k$-shallow cutting of the top $(\leq k)$-level
    and bottom $(\leq k)$-level.
\end{proof}

\section{Smallest disk of all vertices %
   within crossing distance \TPDF{$k$}{k}%
}
\seclab{more:apps}

Let $L$ be a set of lines in the plane. For two points
$\pa, \query \in \R^2$, the \emphi{crossing distance}
$\cross_L(\pa,\query)$ is the number of lines of $L$ intersecting the
segment $\pa\query$.

Given a point $\query \in \R^2$ not lying on any line of $L$, and a
parameter $k$, let
$$S_k(\query) = \Set{\pa \in \VertsX{\ArrX{L}}}{\cross_L(\pa,\query) \leq k}$$
be the set of vertices of $\ArrX{L}$ with crossing distance at most
$k$ from $\query$. The goal is to compute the smallest disk enclosing
all points of $S_k(\query)$, as shown in \figref{cross}.

\begin{figure}[t]
    \centering \includegraphics[scale=0.5]{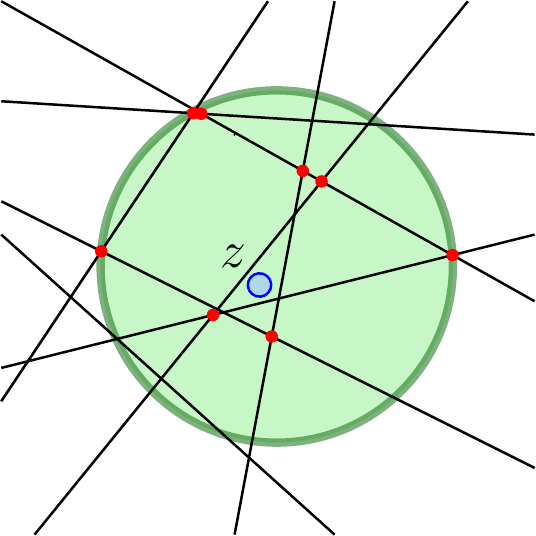}
    \caption{A disk containing all vertices of $\ArrX{L}$ lying within
       crossing distance at most three from $\query$.}
    \figlab{cross}
\end{figure}

\begin{lemma}
    \lemlab{cross:verts} Let $L$ be a set of $n$ lines in the plane
    and let $\query \in \R^2$ be a point not lying on any point of
    $L$.  In $O(n\log n)$ expected time, one can compute the smallest
    disk enclosing all vertices of $\ArrX{L}$ within crossing distance
    at most $k$ from $\query$.
\end{lemma}
\begin{proof}
    When the constraints (points) are explicitly given, this problem
    is \LP{}-type with constant combinatorial dimension. We now apply
    \thmref{l:p:implicit} to obtain an efficient algorithm for this
    problem:
  
    \begin{enumerate}
        \item Each subproblem consists of a simplex $\Cell$, the set
        of lines $L' = L \cap \Cell$, and a number $\uDepth$ which is
        the number of lines of $L$ separating $\Cell$ and
        $\query$.\footnote{A line $\Line$ separates $\Cell$ and
           $\query$ if they lie on opposite sides of $\ell$.} Given a
        disk $\Disk$ defined by the basis, check if there is a vertex
        of $\ArrX{L'}$ which lies outside $\Disk$ and has crossing
        distance at most $k$ from $\query$.

        To this end, compute the zone $\ZoneY{\BX{\Cell}}{L'}$. The
        algorithm chooses a vertex $v$ of $\ZoneY{\BX{\Cell}}{L'}$
        inside $\Cell$ and computes
        $\cross_L(v, \query) = \cross_{L'}(v, \query) + \uDepth$.
        Next, walk around the set of vertices in
        $\ZoneY{\BX{\Cell}}{L'} \cap \Cell$ and compute the crossing
        values using previously computed crossing values.  If at any
        time a vertex of crossing value at most $k$ which is outside
        $\Disk$ is encountered, report that $\Disk$ is an invalid
        solution.

        \item The subproblem $(\Cell, L', \uDepth)$ can be decomposed
        once again using cuttings. Compute a $(1/c)$-cutting (for
        sufficiently large constant $c$) of $L'$ and clip the cutting
        inside $\Cell$.  For each cell $\Cell_i$ in the cutting,
        compute $L' \cap \Cell_i$ and the number of lines separating
        $\Cell_i$ from $\query$.
    \end{enumerate}
  
    The running time of the algorithm is dominated by the running time
    of the violation test, which is proportional to the complexity of
    the zone $\ZoneY{\BX{\Cell}}{L'}$. By \lemref{zone:complexity},
    the violation test runs in $\DX{n} = O(n\log n)$ time.
\end{proof}

\section{Conclusions}
\seclab{conclusions}

Since the conference version of \cite{c-oramt-04}, several
applications of the implicit \LP technique have been found.  For
example, see \cite{AgarwalCSS06,BarbaM19,EppsteinW07,Morin08}.

The natural open problem is to improve the running times for computing
the yolk (and extremal yolk) even further. It seems believable, that
for $d>3$, the log factor in \thmref{discrete:yolk} might not be
necessary. We leave this as an open problem for further research.

\paragraph*{Acknowledgments} %
The first author thanks Stefan Langerman for re-posing the problem of
computing the maximum Tukey depth of a point set at the 2002 Fall
Workshop on Computational Geometry problem session, and for subsequent
discussions.

The authors thank Joachim Gudmundsson for bringing the problem of
computing the yolk to our attention. The last author thanks Sampson
Wong for discussions on computing the yolk in higher dimensions.

Finally, the authors thank the anonymous referees for the detailed
comments and review.

\BibTexMode{%
   \NotSICOMP{%
      \bibliographystyle{alpha}%
   } \SICOMP{%
      \bibliographystyle{siamplain}%
   } \bibliography{yolk}%
}

\BibLatexMode{\printbibliography}

\appendix

\section{Proof sketch of \TPDF{\lemref{sc:alg}}{Lemma
      \ref*{lemma:sc:alg}}}
\apndlab{sc:alg}

Let $H$ be set of $n$ hyperplanes in $\R^d$. We focus on constructing
a $k$-shallow cutting when $d \geq 4$ (for $d < 4$, we can construct
shallow cuttings in $O(n\log n)$ deterministic time
\cite{ct-odasc-16}). The original proof of existence of $k$-shallow
cuttings by \Matousek \cite{m-rph-92} provides a randomized algorithm
for constructing such a cutting.

\smallskip%
At a high level, the approach of \Matousek for constructing a
$k$-shallow cutting is the following:
\begin{compactenumI}
    \item Let $R \subseteq H$ be a random sample of size $n/k$ and
    compute a bottom-vertex triangulation of $\ArrX{R}$. Let $\Cells$
    denote the resulting set of simplices.
    \item Let $\Cells' \subseteq \Cells$ be the subset of simplices
    containing a point of level at most $k$ (with respect to $H$).
    \item For each $\Cell \in \Cells'$, if $\Cell$ intersects $tk$
    hyperplanes of $H$ for some $t > 1$, compute a $(1/t)$-cutting of
    the hyperplanes intersecting $\Cell$ and clip the cutting inside
    $\Cell$. Return this 2-level cutting as the desired $k$-shallow
    cutting.
\end{compactenumI}

\subsection{Computing the top-level cutting} %
The top-level cutting is computed via randomized incremental
construction. The algorithm randomly permutes the hyperplanes of $H$,
label them $h_1, \ldots, h_n$ and let $H_i = \{h_1, \ldots, h_i\}$.
For $i = 1, \ldots, n/k$, the algorithm maintains a collection of
simplices, formed from the arrangement $\ArrX{H_i}$ and which contain
a point of level $k$ (with respect to $H$). Each simplex $\Cell$
maintains pointers to the subset of hyperplanes
$\{h_{i+1}, \ldots, h_n\}$ which intersect $\Cell$ (this is the
\emph{conflict list} of $\Cell$). Each hyperplane $h_j$ for $j > i$
also maintains reverse pointers to the set of simplices it intersects
in the current triangulation. Finally, each cell in the arrangement
maintains the number of hyperplanes of $H$ which lie strictly below
it.

In an update step, insert the hyperplane $h_i$. Using the reverse
pointers, we determine the set of simplices that are split by
inserting $h_i$ into the current arrangement. Using these simplices,
we can find the cells that are split by $h_i$. Fix a cell $C$
intersected by $h_i$ and let $H_C \subseteq H \setminus H_i$ be the
union of the conflict lists over the simplices in $C$.  Suppose $C$ is
split into two new cells $C_1$ and $C_2$. Assume $C_1$ lies above
$h_i$. We determine the number of planes lying below $C_1$ ($C_2$ can
be handled symmetrically).  Let $v$ be a vertex of $C$ lying below
$h_i$. From $v$, we perform a graph search on the boundary of $C$ to
determine the number of hyperplanes of $H_C$ lying strictly below
$C_1$ (adding the number of hyperplanes lying below $C$ to the
count). If at any point this count is greater than $k$, we discard
$C_1$, as it does not cover the $(\leq k)$-level.  This process is
repeated for all cells split by $h_i$. At the end of the process, we
triangulate the newly created cells, and construct the conflict lists
for the new simplices. See \cite[Section 5.4]{dbds-lric-95} for
details on how to efficiently maintain the conflict lists and
arrangement incrementally.

\subsection{Refining the cutting} %
At the end of the process, the algorithm has a collection of simplices
$\Cells$ which cover the $(\leq k)$-level.  For each simplex
$\Cell \in \Cells$, if $\Cell$ has conflict list size $tk$ for some
$t \geq 1$, compute a $(1/t)$-cutting for the hyperplanes intersecting
$\Cell$ and clip the cutting inside $\Cell$.  This ensures that every
simplex in the final two-level cutting intersects at most $k$
hyperplanes of $H$.

\subsection{Analysis sketch} %
In each step of the randomized incremental algorithm, the total amount
of work done is proportional to the size of the conflict lists
destroyed or created. Let $\Cells_i$ denote the current collection of
simplices at step $i$, where $\Cells = \Cells_{n/k}$ is the collection
of cells in the top-level cutting at the end of the process.  We first
analyze the total size of the conflict lists over all simplices in
$\Cells_i$. For each $\Cell \in \Cells_i$, let $w(\Cell)$ be the size
of the conflict list of $\Cell$. For an integer $t \geq 1$, let
$\Cells_i[t] = \Set{\Cell \in \Cells_i}{(t-1)k < w(\Cell) \leq tk}$.
In the original proof of the shallow cutting lemma \Matousek proved
that, roughly speaking, the number of simplices in $\Cells_i$ with
$w(\Cell) \in ((t-1)k, tk] $ is decays exponentially in $t$---formally
$\Ex{\cardin{\Cells_i[t]}} = O(2^{-t} \cardin{\Cells_i})$ \cite[Lemma
2.4]{m-rph-92}.  Using this, one can bound the sum of the conflict
list sizes as (see \cite[Theorem 3]{dbds-lric-95} and \cite[Theorem
8.8]{h-gaa-11}):
\begin{align*}
  \alpha_i := \Ex{\sum_{\Cell \in \Cells_i} w(\Cell)} 
  = O\pth{\cardin{\Cells_i} (n/i)}
  = O\pth{i^{\floor{d/2}} (n/i)}.
\end{align*}
Since the hyperplanes were randomly permuted, we have that that the
amortized work done in the $i$\th step of the algorithm is
$O(\alpha_i/i)$ \cite[Theorem 5]{dbds-lric-95}. As such, the expected
running time to compute the top-level cutting is bounded by:
\begin{align*}
  \sum_{i=1}^{n/k} O\pth{\frac{\alpha_i}{i}}
  = O\pth{\sum_{i=1}^{n/k} \frac{n i^{\floor{d/2}}}{i^2}}
  = O\pth{\frac{n}{k} \cdot n \pth{\frac{n}{k}}^{\floor{d/2}-2}}
  = O\pth{k\pth{\frac{n}{k}}^{\floor{d/2}}},
\end{align*}
where in the second inequality we use the assumption $d \geq 4$.  (For
$d < 4$, the summation solves to $O(n\log(n/k))$.)

As for the second level cutting, fix a simplex $\Cell \in \Cells[t]$
with weight $w(\Cell) \in ((t-1)k, tk]$.  Computing a $(1/t)$-cutting
inside $\Cell$ costs $O(w(\Cell) t^{d-1}) = O(t^d k)$ expected time
\cite{c-chdc-93}.  Thus, the expected running time of the second-level
cutting is bounded by
$O\pth{k \sum_{t=1}^{\infty} \sum_{\Cell \in \Cells[t]} t^d}$.  Again,
by the properties of exponential decay \cite{m-rph-92, dbds-lric-95,
   h-gaa-11} we have that
\begin{align*}
  \Ex{k \sum_{t \geq 1} \sum_{\Cell \in \Cells[t]} t^d} 
  = O\pth{k \pth{n/k}^{\floor{d/2}}}.
\end{align*}

This completes the proof of \lemref{sc:alg}.

\end{document}